\documentclass[a4paper,USenglish,cleveref, autoref, thm-restate]{lipics-v2021}

\pdfoutput=1 %
\hideLIPIcs  %

\graphicspath{{./graphics/}}%

\title{
On the Parameterized Complexity of Capacitated Stack and Queue Layouts with Colored Edges%
} %

\titlerunning{Capacitated Linear Layouts} %

\author{Thomas Depian}{TU Wien, Austria}{tdepian@ac.tuwien.ac.at}{https://orcid.org/0009-0003-7498-6271}{Supported by the Vienna Science and Technology Fund (WWTF) [10.47379/ICT22029].}

\author{Giuseppe Liotta}{University of Perugia, Italy}{giuseppe.liotta@unipg.it}{https://orcid.org/0000-0002-2886-9694}{}

\author{Fabrizio Montecchiani}{University of Perugia, Italy}{fabrizio.montecchiani@unipg.it}{https://orcid.org/0000-0002-0543-8912}{}

\author{Martin N{\"o}llenburg}{TU Wien, Austria}{noellenburg@ac.tuwien.ac.at}{https://orcid.org/0000-0003-0454-3937}{Supported by the Vienna Science and Technology Fund (WWTF) [10.47379/ICT22029].}

\authorrunning{T. Depian, G. Liotta, F. Montecchiani, and M. N{\"o}llenburg} %

\Copyright{Thomas Depian, Giuseppe Liotta, Fabrizio Montecchiani, and Martin N{\"o}llenburg} %



\category{} %

\relatedversion{} %

\acknowledgements{The authors acknowledge the use of AI tools (in particular Claude Opus 5.5) for proofreading the pre-final version of this submission and identifying some additional related literature. No such tools were used for obtaining the results and the writeup of the pre-final version. The authors assume full responsibility for the submission.}%

\nolinenumbers %

\EventEditors{John Q. Open and Joan R. Access}
\EventNoEds{2}
\EventLongTitle{42nd Conference on Very Important Topics (CVIT 2016)}
\EventShortTitle{CVIT 2016}
\EventAcronym{CVIT}
\EventYear{2016}
\EventDate{December 24--27, 2016}
\EventLocation{Little Whinging, United Kingdom}
\EventLogo{}
\SeriesVolume{42}
\ArticleNo{23}

\usepackage{mdframed}
\usepackage{xspace}
\usepackage{todonotes}
\usepackage{complexity}
\usepackage{booktabs}
\usepackage{makecell}
\usepackage{multirow}
\usepackage{mathtools}
\usepackage{pifont}
\usepackage{cite}

\usetikzlibrary{arrows.meta,patterns,ipe}

\newboolean{long}
\newcommand{\AppendixSymbol}{\ding{72}}
\setboolean{long}{true} %
\ifthenelse{\boolean{long}}{
	\NewDocumentEnvironment{prooflater}{m}{\begin{proof}}{\end{proof}\ignorespacesafterend}
	\NewDocumentEnvironment{proofsketch}{o +b}{}{\ignorespacesafterend}
	\newcommand{\restateref}[1]{}
	\NewDocumentEnvironment{statelater}{m}{}{}

	\NewDocumentCommand{\onlyShort}{+m}{}
	\NewDocumentCommand{\onlyLong}{+m}{#1}
	\NewDocumentCommand{\shortLong}{+m +m}{#2}
}{
	\NewDocumentEnvironment{prooflater}{m +b}{%
		\expandafter\global\expandafter\def\csname#1\endcsname{\begin{proof}#2\end{proof}}%
	}{\ignorespacesafterend}
	\NewDocumentEnvironment{proofsketch}{O{Proof sketch.}}{\begin{proof}[#1]}{\end{proof}\ignorespacesafterend}
	\usepackage{apptools}
	\newcommand{\restateref}[1]{[\IfAppendix{\hyperref[#1]{\AppendixSymbol{}}}{\hyperref[#1*]{\AppendixSymbol{}}}]}

	\NewDocumentEnvironment{statelater}{m +b}{%
		\expandafter\global\expandafter\def\csname#1\endcsname{#2}%
	}{\ignorespacesafterend}

	\NewDocumentCommand{\onlyShort}{+m}{#1}
	\NewDocumentCommand{\onlyLong}{+m}{}
	\NewDocumentCommand{\shortLong}{+m +m}{#1}
}
\let\oldrestatable\restatable
\def\restatable{\expandafter\oldrestatable}

\newcommand{\probname}[1]{{\normalfont\textsc{#1}}}
\newcommand{\capacitatedStackTwoPages}{\probname{FoCap2PSl}\xspace}
\newcommand{\capacitatedStackGen}{\probname{FoCapSl}\xspace}
\newcommand{\capacitatedStack}{\capacitatedStackTwoPages}
\newcommand{\capacitatedStackLong}{\probname{Fixed Order Capacitated Stack Layout}\xspace}
\newcommand{\capacitatedStackLongTwoPages}{\probname{Fixed Order Capacitated 2-Page Stack Layout}\xspace}
\newcommand{\capacitatedQueue}{\probname{FoCapQl}\xspace}
\newcommand{\capacitatedQueueLong}{\probname{Fixed Order Capacitated Queue Layout}\xspace}

\newcommand{\capacitatedStackQueueLong}{\probname{Fixed Order Capacitated Stack (Queue) Layout}\xspace}

\newcommand{\probdef}[3]{%
    \begin{mdframed}
		\probname{#1}
		\begin{description}
            \item[Given] #2
            \item[Question] #3
        \end{description}
	\end{mdframed}
}%
\newcommand{\instance}{\ensuremath{\mathcal{I}}\xspace}
\newcommand{\instanceLongStack}{\ensuremath{(G, \delta, \prec, \ell = 2, b_1, b_2)}\xspace}
\NewDocumentCommand{\instanceLongQueue}{O{\ensuremath{\ell}}}{\ensuremath{(G, \delta, \prec, #1, (b_p)_{p \in[#1]})}\xspace}
\newcommand{\instanceLongQueueTwoPages}{\ensuremath{(G, \delta, \prec, \ell = 2, b_1, b_2)}\xspace}

\NewDocumentCommand{\linL}{o o}{\ensuremath{\langle\prec\IfNoValueF{#1}{_{#1}},\sigma\IfNoValueF{#1}{_{#1}}\rangle\IfNoValueF{#2}{_{#2}}}\xspace}
\NewDocumentCommand{\stackL}{o}{\ensuremath{\linL[#1][\mathcal{S}]}\xspace}
\NewDocumentCommand{\queueL}{o}{\ensuremath{\linL[#1][\mathcal{Q}]}\xspace}
\NewDocumentCommand{\stackConflict}{o}{\ensuremath{\mathcal{S}\IfNoValueF{#1}{(#1)}}\xspace}
\NewDocumentCommand{\queueConflict}{o}{\ensuremath{\mathcal{Q}\IfNoValueF{#1}{(#1)}}\xspace}

\NewDocumentCommand{\conComp}{o}{\ensuremath{\lambda\IfNoValueF{#1}{(#1)}}\xspace}
\NewDocumentCommand{\conCompSize}{o}{\ensuremath{\gamma\IfNoValueF{#1}{(#1)}}\xspace}

\newcommand{\col}[1]{\ensuremath{\mathsf{col}(#1)}}

\newcommand{\Size}[1]{\ensuremath{\left\vert #1 \right\vert}}
\newcommand{\BigO}[1]{\ensuremath{\mathcal{O}(#1)}}

\begin{document}

\maketitle

\begin{abstract}
Stack and queue layouts are two extensively studied and complementary linear representations of graphs.
A \emph{stack (queue) layout} $\langle \prec, \sigma \rangle$ of a graph $G$ consists of a total ordering $\prec$ of its vertices and of a partitioning $\sigma$ of its edges into pages, such that no two independent edges on the same page cross (nest).
We propose a generalized model with colored edges (one color per edge) and capacitated pages. Namely, the input graph is equipped with colored edges and a prescribed budget for each page-color pair $(p,c)$ -- this budget determines the maximum number of edges of color $c$ that page $p$ can accept. This model is motivated by fair allocation scenarios where budgets can be used to control the share of monochromatic edges on each page. 
Unfortunately, it is known that deciding whether a stack (queue) layout exists is \NP-hard already with two (one) pages. Thus, we focus on the fixed-order setting -- where the linear order $\prec$ is part of the input -- for which an ordinary 2-page stack or $\ell$-page queue layout can be computed in polynomial time.
Despite the restrictive setting, we show that both capacitated problems are \NP-complete already for two pages.
Furthermore, we study the complexity of the problems when parameterized by the number of pages, colors, and/or the structure of the edge conflict graph with respect to the fixed linear order $\prec$.
Our analysis paints a detailed complexity landscape and uncovers a surprisingly intricate interplay between the considered parameters. 

\end{abstract}

\section{Introduction}
\label{sec:intro}

A \emph{linear layout} $\langle \prec, \sigma \rangle$ of a graph $G$ consists of a total order $\prec$ of its vertices and a partition~$\sigma$ of its edges into pages, such that edges on the same page avoid specified patterns with respect to $\prec$. Over the last four decades, linear layouts have been a central research topic in graph drawing and topological graph theory. Two prominent examples are \emph{stack layouts} and \emph{queue layouts}, in which no two independent edges on the same page are allowed to cross or nest, respectively. These names reflect the data structures used to store the edges of each page during a left-to-right sweep of $\prec$: upon reaching a vertex $v$, edges whose right endpoint is $v$ are removed from the corresponding stack or queue, after which edges whose left endpoint is~$v$ are inserted. With suitable ordering of the operations at each vertex, these accesses follow a last-in-first-out or first-in-first-out discipline, respectively.

Motivated by applications in VLSI design~\cite{CLR87,DBLP:journals/tc/Rosenberg83}, permutation sorting~\cite{DBLP:conf/stoc/Pratt73,DBLP:journals/jacm/Tarjan72}, compact graph encodings~\cite{DBLP:conf/focs/Jacobson89,DBLP:journals/siamcomp/MunroR01}, RNA folding~\cite{HASLINGER1999437}, and traffic-light control~\cite{kainen1990book}, stack and queue layouts have been extensively studied from both algorithmic and structural perspectives. Research has addressed their computational complexity~\cite{DBLP:journals/jgaa/BhoreGMN20,DBLP:journals/jgaa/BhoreGMN22,DBLP:conf/icalp/GanianMOPR24,HR.LOG.1992,DFGS.LLR.2025}, their connections to planarity~\cite{DBLP:journals/algorithmica/AlamBGKP20,DBLP:journals/siamcomp/BekosFGMMRU19,DBLP:journals/jacm/DujmovicJMMUW20}, and their relationships with other graph properties and parameters~\cite{DBLP:journals/siamdm/VandenbusscheWY09,DBLP:journals/combinatorics/Wiechert17}.

Several variants extend this framework by imposing additional constraints or modifying the permitted edge patterns. Merker and Ueckerdt~\cite{DBLP:conf/gd/MerkerU19,DBLP:conf/gd/MerkerU20} introduced the local stack (queue) number of a graph $G$, defined as the minimum $\ell$ such that $G$ admits a stack (queue) layout in which the edges incident to each vertex belong to at most $\ell$ pages. Alam et al.~\cite{DBLP:journals/tcs/AlamBDGKP21} studied dispersable layouts, which require the edges on each page to form a matching. Other works have explored layouts that mix different data structures or use alternative ones to organize the edges on each page (see, e.g.,~\cite{DBLP:journals/tcs/AngeliniBKM22,DBLP:conf/cccg/BekosKPR23,DBLP:conf/wads/GiacomoDFU025,DBLP:conf/stacs/HaunMP25,DBLP:conf/stacs/Katheder0PU25}). More recently, F\"orster et al.~\cite{DBLP:journals/corr/abs-2607-24375} introduced the weighted book thickness of a graph, defined as the minimum average page index of its edges over all stack layouts. Bekos et al.~\cite{DBLP:journals/corr/abs-2607-19968} studied a relaxation that permits a bounded number of otherwise forbidden patterns within each page.

\begin{figure}[b]
	\centering
	\includegraphics[page=1]{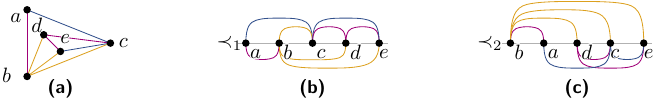}
	\caption{\textbf{\textsf{(a)}} An edge-colored graph $G$ and a \textbf{\textsf{(b)}}  2-page stack and \textbf{\textsf{(c)}}  2-page queue layout of $G$. Observe that there is no 2-page stack layout of $G$ with spine order $\prec_1$ where the two blue edges are on different pages.}
	\label{fig:intro}
\end{figure}

\subparagraph*{Capacitated Linear Layouts.} In this paper, we introduce \emph{capacitated linear layouts}, a generalization of stack and queue layouts in which each edge of the input graph $G$ is assigned one of $\chi$ colors, and each page has a separate budget for each color. The goal is to find a stack or queue layout using at most $\ell$ pages, subject to the constraint that, on each page, the number of edges of each color does not exceed the corresponding budget. The classical notions of stack and queue layouts are recovered by setting $\chi=1$ and assigning every page a budget of $|E(G)|$; see \Cref{fig:intro} for an illustrating example. %

\medskip\noindent\emph{Motivating Examples.} 
In \emph{partitioned book embeddings}, the assignment of edges to pages is part of the input~\cite{DBLP:conf/compgeom/LozzoFR26,ADBDB.IP2.2013}.
These can be seen as a special case of capacitated linear layouts, where edges assigned to page $p$ receive color $p$, and the budget of page $p$ is $\Size{E(G)}$ for color $p$ and $0$ for every other color.
When colors encode different types of elements---for example, in signed graphs with positive and negative relations among social media users~\cite{DBLP:conf/chi/LeskovecHK10}---capacitated linear layouts can model fair allocation problems, with suitable budgets specifying the share of each type assigned to each page. 
Moreover, edge colors can distinguish job classes or customers, and per-color page budgets constrain the number of jobs of each class assigned to each buffer, similar to the scheduling problem considered by Das and Wiese~\cite{DW.MMW.2017}.
In graph visualization, users may focus their attention on the first few pages, motivating layouts that distribute edges of each color evenly across pages to ensure that edges from every color appear among the first pages (see~\cite{DBLP:journals/isci/EadesHLMNPW25} for related fairness-aware formulations in graph drawing).
Finally, fairness constraints have also been studied in other graph problems, including graph clustering~\cite{DBLP:conf/aaai/FroeseKN22} and graph bisection~\cite{DBLP:conf/esa/0002L0S23}.

\subparagraph*{Contributions.}
We study the computational complexity of capacitated linear layouts from a parameterized perspective. Deciding whether a graph admits a 1-page queue layout or a 2-page stack layout is already \NP-complete~\cite{HR.LOG.1992,Wid.CHC.1982,BK.BTG.1979}. These hardness results carry over to our model, implying para\NP-hardness with respect to the number of pages, the number of colors, and their combination. We, therefore, focus on the fixed-order setting in which the vertex order~$\prec$ is given as part of the input and must be preserved. This restriction is consistent with many of the above-discussed motivating applications and indeed is well studied in the literature on stack and queue layouts~\cite{DBLP:conf/gd/BhoreGMN19,DBLP:journals/jgaa/BhoreGMN20,DBLP:journals/dmtcs/DujmovicW04,HR.LOG.1992,DBLP:conf/cocoa/LiuCH20,DBLP:conf/stacs/Unger88}. In particular, when $\prec$ is fixed, deciding if a graph admits an $\ell$-page queue layout can be done in polynomial time for any $\ell>0$~\cite{HR.LOG.1992}. For stack layouts the problem is tractable for $\ell \leq 2$, while it is \NP-hard for $\ell \geq 4$ pages~\cite{DBLP:conf/stacs/Unger88} and open for three pages~\cite{DBLP:journals/jgaa/BachmannRS24,DBLP:conf/sosa/SGLS26}. 
We formalize the following problems.

\probdef{
	\capacitatedStackQueueLong [{\normalfont short} \capacitatedStackGen (\capacitatedQueue)]}
{A graph $G = (V, E)$, an edge coloring function $\delta \colon E \rightarrow [\chi]$ on $\chi$ colors, a linear order $\prec$ on $V$, an integer $\ell \in \mathbb{N}$, and budgets $b_p \colon [\chi] \to \mathbb{N}_0$ for every page $p \in [\ell]$.}{
	Does there exist an $\ell$-page stack (queue) layout \linL of $G$ such that $\Size{\{e \in E \mid \delta(e) = c,\ \sigma(e) = p\}} \leq b_p(c)$ for every page $p \in [\ell]$ and color $c \in [\chi]$?
}

Since we mainly consider $\ell = 2$ pages for \capacitatedStackGen, we call the corresponding restricted variant \capacitatedStackLongTwoPages (\capacitatedStack for short).
We note that our problem parallels previous studies of bounded vertex colorings, where each color can be used at most $b$ times, for some value of $b$; see, e.g., the work by Jansen~\cite{Jan.mes.2003}.
In our setting, the graph to color would be the edge conflict graph with respect to the linear order $\prec$ and the number of allowed colors is the number of pages $\ell$.
However, recall that we have, on the one hand, an edge-colored graph $G$ and we can observe, on the other hand, that known hardness results for bounded vertex colorings for constant $b$ require an unbounded number of colors, i.e., pages.
In contrast, when adding edge colors we show that, despite~$\prec$ being fixed, both problems are \NP-complete already for two pages ($\ell=2$), which clearly separates the complexity of the problems with and without page capacities unless $\P= \NP$. Furthermore, we study the complexity of the problems when parameterized by the number of pages, colors, and/or the structure of the edge conflict graph.
We summarize our results in \Cref{fig:overview}. %

\begin{figure}
	\centering
	\input{graphics/cappar_new}
	\caption{Parameterized complexity of \textbf{\textsf{(a)}} \capacitatedStack (\Cref{sec:stack-layouts}) and \textbf{\textsf{(b)}} \capacitatedQueue (\Cref{sec:queue-layouts}) when parameterized by the number of pages $\ell$, the number of colors $\chi$, the number of connected components $\conComp$ in the conflict graph, and/or their maximum size $\conCompSize$.
		The result for \capacitatedQueue marked with $*$ applies only if $\ell = 2$.
		Moreover, note that $\ell = 2$ holds for every instance of \capacitatedStack and, therefore, we omit the parameter in (a).}
	\label{fig:overview}
\end{figure}

\onlyShort{
\smallskip
\noindent
\emph{Full details for statements marked by \AppendixSymbol\ can be found in the appendix.}
}

\section{Preliminaries}
\label{sec:preliminaries}
For an integer $\ell \geq 1$, we let $[\ell]$ denote the set $\{1, 2, \ldots, \ell\}$ and define $[\ell]_0 \coloneqq [\ell] \cup \{0\}$.	
Let $G = (V, E)$ be a simple and undirected graph.
We let $\mathcal{C}(G)$ denote the set of connected components of $G$.
Let $\prec$ be a linear order of the vertices $V$.
We say that two edges $uv, xy \in E$ with $u \prec v$, $x \prec y$, and $u \prec x$ \emph{cross} under $\prec$ if $u \prec x \prec v \prec y$ and they \emph{nest} under $\prec$ if $u \prec x \prec y \prec v$.
For a given integer $\ell$, let $\sigma \colon E \to [\ell]$ be a function that assigns each edge to one of $\ell$ \emph{pages}.
We call the \emph{linear layout} \linL a \emph{stack} (\emph{queue}) \emph{layout} if there are no two edges $uv, xy \in E$ with $\sigma(uv) = \sigma(xy)$ that cross (nest) under $\prec$.
In this case, we also refer to $\prec$ as the \emph{spine order} and to $\sigma$ as the \emph{page assignment} of \linL.
The \emph{conflict graph} $\stackConflict[G, \prec] = (E, F)$ of $G$ and $\prec$ for stack layouts contains one vertex for every edge $e \in E$.
Two edges $e_1, e_2 \in E$ are adjacent in $\stackConflict[G, \prec]$ if and only if $e_1$ and $e_2$ cross under $\prec$; see \Cref{fig:example-conflict-graph} for an example.
We write \stackConflict if $G$ and $\prec$ are clear from the context and define the conflict graph $\queueConflict[G, \prec]$ for queue layouts analogously with conflicts defined by nesting edges.
Observe that there is a one-to-one correspondence between page assignments of $\ell$-page stack (queue) layouts \linL and proper $\ell$-vertex colorings of the conflict graph.

For a graph $G = (V, E)$, let $\delta\colon E \to [\chi]$ be a function that assigns each edge $e \in E$ one of $\chi \leq m$ \emph{colors}.
We call $G = (V, E, \delta)$ an \emph{edge-colored} graph.
For an instance $\instance = \instanceLongQueue$ of \capacitatedStackGen (\capacitatedQueue), we define $n \coloneqq \Size{V}$, $m \coloneqq \Size{E}$, $\chi$ to be the number of colors, $\conComp \coloneqq \Size{\mathcal{C}(\mathcal{H})}$, and $\conCompSize \coloneqq \max\{\Size{V(C)} \mid C \in \mathcal{C}(\mathcal{H})\}$, where $\mathcal{H}$ denotes the corresponding conflict graph for stack (queue) layouts.
Moreover, we call $b_p(c)$ the \emph{budget} of color $c \in [\chi]$ on page $p \in [\ell]$.
We say that the instance \instance is \emph{balanced} if $b_p(c) = b_q(c)$ holds for every two pages $p,q \in [\ell]$ and color $c \in [\chi]$.
It is \emph{tight} if $\sum_{p \in [\ell],c \in [\chi]}b_p(c) = m$, i.e., the sum over all budgets equals the number of edges.
A stack (queue) layout \linL which respects the budget constraints imposed by \instance is called \emph{admissible} and constitutes a \emph{solution} to \instance.\footnote{Note that in principle the page assignment $\sigma$ would be sufficient as a witness for the existence of a solution since the spine order $\prec$ is part of the input.}

\begin{figure}
	\centering
	\includegraphics[page=1]{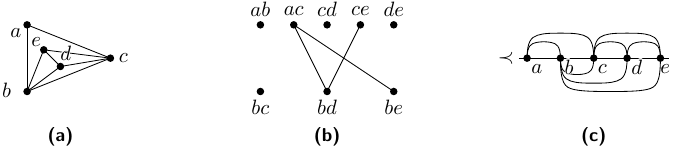}
	\caption{\textbf{\textsf{(a)}} An example graph $G$ and \textbf{\textsf{(b)}} its conflict graph $\stackConflict$ for the spine order $a \prec b \prec c \prec d \prec e$. Since $\stackConflict$ is bipartite, there exists a 2-page stack layout $\linL$.}
	\label{fig:example-conflict-graph}
\end{figure}

\section{Fixed Order 2-Page Stack Layouts}
\label{sec:stack-layouts}
We start our investigation of capacitated linear layouts by considering 2-page stack layouts.
As already noted in earlier work~\cite{DBLP:journals/jgaa/BhoreGMN20}, for a fixed spine order $\prec$, a graph $G$ has a 2-page stack layout $\linL$ if and only if the conflict graph $\stackConflict[G, \prec]$ is bipartite.
When pages have capacities, the situation becomes more intricate.
In particular, in the upcoming sections we will uncover an interplay between the structure of $\stackConflict$ and the number of colors $\chi$.

\subsection{Parameterizing by the Number of Connected Components of $\boldsymbol{\stackConflict}$}
\label{sec:stack-layouts-components}
Let $\instance = \instanceLongStack$ be an instance of \capacitatedStack.
As our first parameter, we consider the number $\conComp$ of connected components in $\stackConflict$.
As a ``warm-up'' exercise, we present a fixed-parameter algorithm in $\conComp$.
Later, in \Cref{sec:stack-layouts-components-size}, we justify this parameter choice by establishing \NP-hardness for unbounded $\conComp$.

\begin{restatable}\restateref{thm:stack-layouts-components-fpt}{theorem}{theoremStackComponentsFPT}
	\label{thm:stack-layouts-components-fpt}
    \capacitatedStackLongTwoPages can be solved constructively in $\BigO{2^{\conComp}\cdot (n + \chi) + n^2}$ time (which is \FPT\ in $\conComp = \Size{\mathcal{C}(\stackConflict[G, \prec])}$).
\end{restatable}
\begin{proofsketch}
    The graph $G$ must be planar to admit a 2-page stack layout, so we can compute $\stackConflict$ in $\BigO{n^2}$ time.
	Let $C_1, \ldots, C_{\conComp}$ be the connected components of $\stackConflict$.
    Every connected component $C_i$, $i \in [\conComp]$, must be bipartite, otherwise we are facing a negative instance.
    Within~$C_i$, edges in the same bipartition must be assigned to the same page.
    The chosen page assignment for $C_i$ is independent of the choice made for other connected components~$C_j$, $j \in [\conComp]$, $i \neq j$, which allows us to branch over all $2^{\conComp}$ possible page assignments.
\end{proofsketch}
\begin{prooflater}{ptheoremStackComponentsFPT}
	Let $\instance = \instanceLongStack$ be an instance of \capacitatedStack.
	Observe that $G$ must be planar to admit a 2-page stack layout, which we can check in $\BigO{n}$ time~\cite{HT.EPT.1974}.
	Thus, for the remainder of the proof we use the fact that $m \in \BigO{n}$ holds for any positive instance; instances where $G$ is not planar are rejected.
	
	We first compute, in $\BigO{n^2}$ time, the conflict graph $\stackConflict$ and check if it is bipartite; if not, we reject \instance as this is a necessary condition to admit a solution.
	Let $C_1, \ldots, C_{\conComp}$ be the connected components of $\stackConflict$.
	Consider a connected component $C_i = (E_i, F_i)$, $i \in [\conComp]$.
	Its bipartition into $E_i^1$ and $E_i^2$ is unique and encodes two possible page assignments for the edges of $G$ contained in $C_i$: Either $\sigma(e) = 1$ if $e \in E_i^1$ and $\sigma(e) = 2$ if $e \in E_i^2$ or vice-versa.
	Note that this choice is independent of the choice made for other connected components~$C_j$, $j \in [\conComp]$, $i \neq j$.
	Thus, we can branch on all $2^{\conComp}$ possible page assignments for the edges of $G$.
	Each branch corresponds to a 2-page stack layout $\linL$, and we can check in $\BigO{n + \chi}$ time if it is admissible.
	The instance \instance admits a solution if and only if one branch corresponds to an admissible layout.
	The running time follows by combining all steps.
\end{prooflater}

\subsection{Hardness for Many Constant-Sized Components in $\boldsymbol{\stackConflict}$}
\label{sec:stack-layouts-components-size}
In this section, we justify the use of $\conComp = \Size{\mathcal{C}(\stackConflict)}$ as parameter by establishing \NP-completeness of \capacitatedStack for unbounded $\conComp$.
Furthermore, our reduction shows that the problem remains \NP-complete even when every connected component of $\stackConflict$ has constant size.
To this end, we reduce from \probname{Monotone Not All Equal 3Sat-E4}, a \probname{3-Sat} variant that remains \NP-complete~\cite{DD.shv.2020}.
The input to this problem is a formula $\varphi$ over variables $X$ with clause set~$C$ such that each clause $c \in C$ consists of exactly three distinct positive literals and every variable $x \in X$ appears in exactly four clauses.
The question is whether there exists a truth assignment $\Psi\colon X \to \{0,1\}$ such that every clause contains two variables with different truth assignments, i.e., one that is assigned to true and one that is assigned to false.

\newcommand{\stateStackVariableGadgetObservation}{
\begin{observation}
	\label{obs:stack-layouts-components-size-variable}
	In every admissible 2-page stack layout \linL of the variable gadget for a variable $x \in X$ we have $\sigma(s_xt_x) \neq \sigma(u_xv_x) = \sigma(a_xb_x^i)$ for every $i = 1,2,3,4$.
\end{observation}
}

\subparagraph*{Overview of the Reduction.}
The high-level idea of our reduction is as follows.
We construct an instance \instance where we introduce for every variable $x \in X$ eight edges of a unique color $\col{x}$.
Four of these edges\onlyShort{, $a_xb_x^i$, $i = 1,2,3,4$} are placed in the gadget for the variable $x$\onlyShort{; see also \Cref{fig:stack-layout-component-size-gadgets}a}.
\shortLong{%
Thanks to the edge $s_xt_x$, which has a different color $d_x$, these edges are forced to be placed on the same page and will encode the variable assignment of $x$.
We set $b_1(d_x) = b_2(d_x) = 1$ and use an additional edge $u_xv_x$ of color $d_x$ to guarantee that an edge of color $d_x$ will be placed on both pages. 
We can make the following useful observation:
\stateStackVariableGadgetObservation
}{
These edges are forced to be placed on the same page and will encode the truth assignment $\Psi(x)$.
}

Each of the remaining four edges is part of a clause gadget for a clause $c \in C$ with $x \in c$; recall that $x$ appears in exactly four clauses.
\shortLong{%
The clause gadget for $c$ consists of three  edges $a_{q}^cb_{q}^c$, one for each literal $q \in c$ which is colored $\col{q}$; recall that $\varphi$ is monotone.
We set $b_1(\col{x}) = b_2(\col{x}) = 4$ and by \Cref{obs:stack-layouts-components-size-variable}, this transports the variable assignment from the variable gadgets into the clause gadgets.
To validate an assignment, we introduce three edges $u_q^cv_q^c$ of color $\col{c}$, one for each literal $q \in c$.
The edges are placed on the spine such that $a_{q}^cb_{q}^c$ and $u_q^cv_q^c$ must be on different pages; see also \Cref{fig:stack-layout-component-size-gadgets}b.
We set $b_1(\col{c}) = b_2(\col{c}) = 2$, which allows us to relate solutions to $\varphi$ and solutions to \instance:
}{
The variable gadget enforces four of eight edges to be on the same page.
We can set the budget for $\col{x}$ such that the four remaining edges must be on the same page independent of their location along the spine.
This allows us to transport the variable assignment from the variable gadgets into the clause gadgets.
Finally, the gadget for every clause $c \in C$ also contains four disjoint edges of a unique color $\col{c}$.
Three of these edges are in conflict with the edges for the variables, each of them with an edge of a different variable $x \in c$.
The last edge has no conflict.
By setting the budget of $\col{c}$ accordingly, we can ensure that at most two variable-edges in the clause are on the same page, i.e., for every clause $c$ there is one variable set to true and one to false.
In the following, we describe the gadgets in detail.
}

\begin{figure}
	\centering
	\includegraphics[page=1]{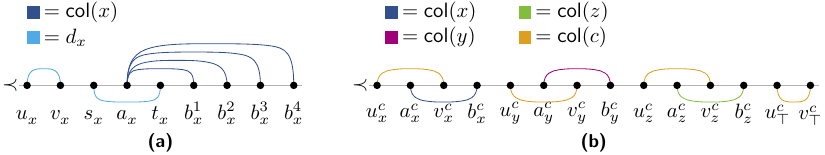}
	\caption{Gadgets representing \textbf{\textsf{(a)}} the variable $x$ in its true state and \textbf{\textsf{(b)}} the clause $c = (x \lor y \lor z)$ where $x$ and $z$ are set to true and $y$ to false in the reduction for \Cref{thm:stack-layouts-components-size}.}
	\label{fig:stack-layout-component-size-gadgets}
\end{figure}

\begin{statelater}{stackComponentSizeGadgets}
\subparagraph*{Variable Gadget.}
Let $x \in X$ be a variable of $\varphi$.
The variable gadget for $x$ consists of a \emph{selecting} edge $s_xt_x$ and a \emph{dummy} edge $u_xv_x$.
Moreover, it consists of four \emph{variable} edges $a_xb_x^i$, $i = 1,2,3,4$.
The order of the endpoints of these edges is $u_x \prec v_x \prec s_x \prec a_x \prec t_x \prec b_x^1 \prec b_x^2 \prec b_x^3 \prec b_x^4$; see also \Cref{fig:stack-layout-component-size-gadgets}a.
All edges of the form $a_xb_x^i$, $i = 1,2,3,4$, receive the same color $\col{x}$.
The two other edges receive the same color which is different from $\col{x}$, i.e.,
$\delta(s_xt_x) = \delta(u_xv_x) = d_x \neq \col{x}$.
We distribute the budget for the colors on the pages as follows.
For color $d_x$, we set $b_1(d_x) = b_2(d_x) = 1$, and for color $\col{x}$, we set $b_1(\col{x}) = b_2(\col{x}) = 4$.
The edges $s_xt_x$ and $u_xv_x$ will be the only two edges of color $d_x$.
\onlyLong{
Together with their placement with respect to the variable edges, we can make the following observation; recall also \Cref{fig:stack-layout-component-size-gadgets}a.
\stateStackVariableGadgetObservation
}

\subparagraph*{Clause Gadget.}
Let $c \in C$ be a clause of $\varphi$.
Since $\varphi$ is monotone, we have that $c = (x \lor y \lor z)$ for some $x,y,z \in X$.
The clause gadget for $c$ consists of three \emph{literal} edges $a_{x}^cb_{x}^c$, $a_{y}^cb_{y}^c$, and~$a_z^cb_{z}^c$.
Moreover, it contains four \emph{validator} edges $u_x^cv_x^c$, $u_{y}^cv_{y}^c$, $u_{z}^cv_{z}^c$, and $u_{\top}^cv_{\top}^c$.
We use the order $u_{q}^c \prec a_{q}^c \prec v_{q}^c \prec b_{q}^c$ among the endpoints of the literal and validator edge corresponding to the literal $q \in \{x,y,z\}$.
In the end, we concatenate the respective total orders and append $u_{\top}^c \prec v_{\top}^c$; see \Cref{fig:stack-layout-component-size-gadgets}b.
We re-use the color of a variable for its corresponding literal edge.
More formally, we have $\delta({a_{q}^cb_{q}^c}) = \col{q}$ for $q \in \{x, y, z\}$.
Moreover, we color the validator edges in a unique color $\col{c}$, i.e., $\delta(u_{q}^cv_{q}^c) = \col{c}$, for $q \in \{x,y,z,\top\}$, and set the budget $b_1(\col{c}) = b_2(\col{c}) = 2$.
The budget constraint together with the relative placement of these edges allows us to establish the following lemma, which will play a central role when establishing correctness of the reduction.    
\end{statelater}

\begin{restatable}\restateref{lem:stack-layouts-components-size-clause}{lemma}{lemmaStackComponentsSizeClause}
	\label{lem:stack-layouts-components-size-clause}
	Let \linL be an admissible 2-page stack layout of the constructed instance that contains the clause gadget for a clause $c \in C$ with $c = (x \lor y \lor z)$.
	We have $\sigma(a_{q}^cb_{q}^c) = \sigma(s_qt_q)$ for every $q \in \{x,y,z\}$.
	Moreover, there exist two literals $p,q \in \{x,y,z\}$ such that $\sigma(a_{p}^cb_{p}^c) \neq \sigma(a_{q}^cb_{q}^c)$.
\end{restatable}
\begin{proofsketch}
	Let $c \in C$ be a clause of $\varphi$ with $c = (x \lor y \lor z)$.
	For the first statement, assume that there exists a $q \in \{x,y,z\}$ such that $\sigma(a_{q}^cb_{q}^c) \neq \sigma(s_qt_q)$.
	By \Cref{obs:stack-layouts-components-size-variable}, we have that $\sigma(a_qb_q^i) \neq \sigma(s_qt_q)$ for every $i = 1,2,3,4$.
	Thus, there are at least five edges of color $\col{q}$ on the same page $\sigma(a_{q}^cb_{q}^c)$, which is a contradiction.
	For the second statement, assume that there exists an admissible 2-page stack layout with $\sigma(a_{x}^cb_{x}^c) = \sigma(a_{y}^cb_{y}^c) = \sigma(a_{z}^cb_{z}^c)$.
	From $u_{q}^c \prec a_{q}^c \prec v_{q}^c \prec b_{q}^c$ it follows that $\sigma(a_{q}^cb_{q}^c) \neq \sigma(u_{q}^cv_{q}^c)$ for every $q \in \{x,y,z\}$ and hence $\sigma(u_{x}^cv_{x}^c) = \sigma(u_{y}^cv_{y}^c) = \sigma(u_{z}^cv_{z}^c)$.
	Thus, there exists a page that contains at least three edges of color $\col{c}$, which is a contradiction.
\end{proofsketch}
\begin{prooflater}{plemmaStackComponentsSizeClause}
	Let $c \in C$ be a clause of $\varphi$ with $c = (x \lor y \lor z)$.
	For the first statement, assume for the sake of a contradiction that there exists a $q \in \{x,y,z\}$ such that $\sigma(a_{q}^cb_{q}^c) \neq \sigma(s_qt_q)$.
	By \Cref{obs:stack-layouts-components-size-variable}, we have that $\sigma(a_qb_q^i) \neq \sigma(s_qt_q)$ for every $i = 1,2,3,4$.
	Thus, there are at least five edges of color $\col{q}$ on the same page $\sigma(a_{q}^cb_{q}^c)$.
	However, we have $b_1(\col{q}) = b_2(\col{q}) = 4$.
	Consequently, \linL cannot be an admissible solution.
	 
	To show that also the second statement is true, assume for the sake of a contradiction that there exists an admissible 2-page stack layout with $\sigma(a_{x}^cb_{x}^c) = \sigma(a_{y}^cb_{y}^c) = \sigma(a_{z}^cb_{z}^c)$.
	From $u_{q}^c \prec a_{q}^c \prec v_{q}^c \prec b_{q}^c$ it follows that $\sigma(a_{q}^cb_{q}^c) \neq \sigma(u_{q}^cv_{q}^c)$ for every $q \in \{x,y,z\}$ and hence $\sigma(u_{x}^cv_{x}^c) = \sigma(u_{y}^cv_{y}^c) = \sigma(u_{z}^cv_{z}^c)$.
	Recall that all validator edges have color $\col{c}$.
	Thus, there exists a page that contains at least three edges of color $\col{c}$.
	However, we have $b_1(\col{c}) = b_2(\col{c}) = 2 < 3$, which contradicts the assumption that $\linL$ was an admissible solution.
	Consequently, $\sigma(a_{x}^cb_{x}^c) = \sigma(a_{y}^cb_{y}^c) = \sigma(a_{z}^cb_{z}^c)$ cannot hold, and we can find $p,q \in c$ that fulfills the second claim from the statement.
\end{prooflater}

\shortLong{
 As for the variable gadget, we add a fourth edge $u_\top^cv_\top^c$ of color $\col{c}$ to guarantee that every page receives precisely two edges of color $\col{c}$.
 We observe that the variable gadgets induce the largest connected components of $\stackConflict$; recall \Cref{fig:stack-layout-component-size-gadgets} and use \Cref{lem:stack-layouts-components-size-clause} to show:
}{
\subparagraph*{Putting Everything Together.}
For an instance $\varphi$ of \probname{Monotone Not All Equal 3Sat-E4}, we arbitrarily order the variables $X$ and clauses $C$.
We construct the variable gadget and clause gadget for every variable $x \in X$ and clause $c \in C$.
Finally, we order them on the spine $\prec$ by first placing the variable gadgets in increasing order, followed by the clause gadgets, again in increasing order.
We observe that the variable gadgets induce the largest connected components of $\stackConflict$; recall \Cref{fig:stack-layout-component-size-gadgets}.
Using \Cref{lem:stack-layouts-components-size-clause}, we obtain:
}
\begin{restatable}\restateref{thm:stack-layouts-components-size}{theorem}{theoremStackComponentsSize}
	\label{thm:stack-layouts-components-size}
	\capacitatedStackLongTwoPages is \NP-complete even when restricted to balanced and tight instances $\instanceLongStack$ 
    where every connected component of $\stackConflict[G,\prec]$ has at most five vertices.
\end{restatable}
\begin{prooflater}{ptheoremStackComponentsSize}
	Containment in \NP\ is immediate.
	To show \NP-hardness, let $\varphi$ be an instance of \probname{Monotone Not All Equal 3Sat-E4} on $\eta$ variables $X$ and $\mu$ clauses $C$, and let $\instance = \instanceLongStack$ be the instance of \capacitatedStack obtained from the above-described reduction.
	Since every gadget has constant size and is constructed precisely once for every variable and clause of $\varphi$, we can conclude that the size of $\instance$ is in $\BigO{\eta + \mu}$.
	Moreover, the number of colors is also linear in $\eta + \mu$.
	The instance can be constructed in the same time.
	Regarding the size of the connected components in $\stackConflict$, we can observe that every variable gadget induces two connected components:
	One corresponds to the single edge $u_xv_x$ which is in conflict with no other edge.
	The other edges of the variable gadget induce a single connected component of size five in $\stackConflict$; recall \Cref{fig:stack-layout-component-size-gadgets}a.
	Clause gadgets induce connected components of size at most two in $\stackConflict$.
	Since all gadgets are disjoint with respect to $\prec$, the largest connected component of $\stackConflict$ has size at most five.
	We now establish correctness of the reduction.
	
	\proofsubparagraph{($\boldsymbol{\Rightarrow}$)}
	Assume that $\varphi$ is a positive instance and let $\Psi\colon X \to \{0,1\}$ be a witnessing truth assignment.
	We now construct based on $\Psi$ a page assignment $\sigma$ such that $\linL$ witnesses that \instance is a positive instance of \capacitatedStack.
	To this end, consider a variable $x \in X$ and its corresponding variable gadget.
	If $x$ is set to false, we assign the edge $s_xt_x$ to the first page and otherwise to the second page.
	The remaining edges are assigned to the respective other page (as also prescribed by the necessary page assignment from \Cref{obs:stack-layouts-components-size-variable}).
	This assignment is crossing-free since the variable edges are only in conflict with the selecting edge $s_xt_x$; the dummy edge $u_xv_x$ does not cross any other edge.
	Moreover, we assigned to each page precisely one edge of color $d_x$, which matches the budget for these pages.
	
	We apply the ``mirrored'' page assignment in the clause gadgets.
	More formally, for a clause $c = (x \lor y \lor z) \in C$, we assign the literal edge $a_{q}^cb_{q}^c$ to the page $\sigma(s_qt_q)$ for every $q \in \{x,y,z\}$.
	The validator edge $u_{q}^cv_{q}^c$ is assigned to the respective other page, i.e., it holds $\sigma(u_{q}^cv_{q}^c) \neq \sigma(a_{q}^cb_{q}^c)$.
	This ensures that the clause gadget will be crossing-free.
	For the last validator edge, i.e., $u_{\top}^cv_{\top}^c$, we consider the unique variable $q \in \{x,y,z\}$ whose truth assignment is different from the other two.
	We set $\sigma(u_{\top}^cv_{\top}^c) = \sigma(u_{q}^cv_{q}^c)$, i.e., assign $u_{\top}^cv_{\top}^c$ to the same page as $u_{q}^cv_{q}^c$.
	This way each page contains two edges of color $\col{c}$, which matches the budget for color $\col{c}$.
	Finally, regarding the color $\col{q}$ for some $q \in X$, we make the following observation.
	Since $\sigma(s_qt_q) = \sigma(a_{q}^cb_{q}^c)$, we conclude $\sigma(a_qb_q^i) \neq \sigma(a_{q}^cb_{q}^c)$ for every $i = 1,2,3,4$.
	Since every variable appears in exactly four clauses, we derive that every page contains four edges of color $\col{q}$, and we also have $b_1(\col{q}) = b_2(\col{q}) = 4$.
	Combining all, we conclude that \linL is a solution to \instance.
	
	\proofsubparagraph{($\boldsymbol{\Leftarrow}$)}
	Assume that \instance is a positive instance of \capacitatedStack and let \linL be a witnessing solution.
	We now construct a variable assignment $\Psi\colon X \to \{0,1\}$ that satisfies $\varphi$.
	To this end, consider a variable $x \in X$.
	We set it to true if and only if $\sigma(s_xt_x) = 2$.
	To see that $\Psi$ is a solution to $\varphi$, consider an arbitrary clause $c = (x \lor y \lor z)\in C$.
	Since \linL is a solution, we know from \Cref{lem:stack-layouts-components-size-clause} that there exist two literals $p,q\in \{x,y,z\}$ such that $\sigma(a_{p}^cb_{p}^c) \neq \sigma(a_{q}^cb_{q}^c)$.
	Since the very same lemma guarantees $\sigma(a_{q}^cb_{q}^c) = \sigma(s_qt_q)$ for every $q \in \{x,y,z\}$, we conclude that there are two variables $p,q \in \{x, y, z\}$ with $\Psi(p) \neq \Psi(q)$.
	These two variables NAE-satisfy the clause $c$ and since $c$ was selected arbitrarily, we derive that $\Psi$ is a solution to $\varphi$.
\end{prooflater}

\subsection{Parameterizing by the Number of Colors}
\label{sec:stack-layouts-colors}
In this section, we consider the number $\chi$ of colors as parameter.
In particular, we show with \Cref{thm:stack-layouts-colors} that \capacitatedStackLongTwoPages is in \XP\ and \W[1]-hard with respect to this parameter.
We establish \XP-membership in \Cref{lem:stack-layouts-colors-xp} and provide the complementing hardness result in \Cref{lem:stack-layouts-colors-w-one}.

\subparagraph*{\textsf{XP}-Membership.}
Let $\instance = \instanceLongStack$ be an instance of \capacitatedStack.
We provide a dynamic programming (DP) algorithm to establish \XP-membership.
At a high-level, we iteratively consider the connected components of $\stackConflict$.
Using a $\BigO{\conComp \cdot (m + 1)^{\chi}}$-sized DP-table~$D$, we track for every \emph{color vector} $\overline{\mu} = (\mu_1, \ldots, \mu_{\chi})$ whether it is possible to have $\mu_c \in [m]_0$ edges of color $c \in [\chi]$ on the first page using the first $i \in [\conComp]_0$ connected components of~\stackConflict.
The respective number of colored edges on the second page can be obtained by simple arithmetic.
We check in the end whether a vector matching our budgets can be realized.

\begin{restatable}\restateref{lem:stack-layouts-colors-xp}{lemma}{lemmaStackColorsXP}
	\label{lem:stack-layouts-colors-xp}
    \capacitatedStackLongTwoPages can be solved constructively in $\BigO{\chi \cdot \conComp \cdot (m + 1)^{\chi} + n^2}$ time (which is \XP\ in the number of colors $\chi$).
\end{restatable}
\begin{figure}
	\centering
	\includegraphics[page=1]{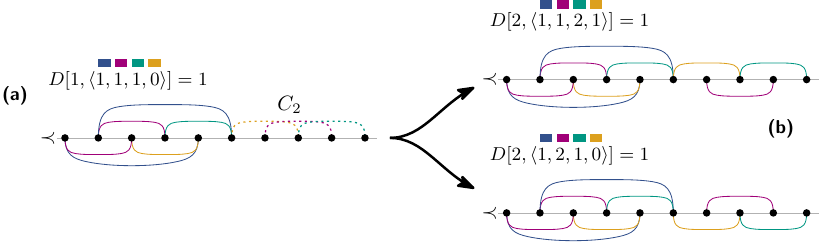}
	\caption{\textbf{\textsf{(a)}} An admissible 2-page stack layout for the edges $E(C_1)$. The edges of $E(C_2)$ are indicated with dashed lines. \textbf{\textsf{(b)}} The two possible solutions that we can obtain when adding $E(C_2)$. We indicate above the layout the corresponding color vectors and table entries.}
	\label{fig:stack-layout-dp-transition}
\end{figure}
\begin{proofsketch}
As $G$ must be planar to admit a solution, we can compute the conflict graph~$\stackConflict$ in $\BigO{n^2}$ time.
Consider an arbitrary order $C_1, \ldots, C_{\conComp}$ of its connected components.
Let $i \in [\conComp]$ be an index and let $\overline{\mu} = (\mu_1, \ldots, \mu_{\chi})$ be a color vector with $\mu_c \in [m]_0$ for every color $c \in [\chi]$.
We set $D[i, \overline{\mu}] = 1$ if and only if there exists a layout $\langle\prec, \sigma_i\rangle$ for the edges $\bigcup_{j \in [i]} E(C_j)$ with $\mu_c$ edges of color $c$ on the first page, for every $c \in [\chi]$, i.e., $\langle\prec, \sigma_i\rangle$ \emph{induces}~$\overline{\mu}$.
As the base case, we have $D[0, \overline{\mu}] = 1$ if and only if $\overline{\mu} = \langle 0, \ldots, 0\rangle$.
For general $i \geq 1$, we observe that $C_i$ must be bipartite (otherwise, \instance is a negative instance).
Thus, there are two page assignments for the edges represented by $C_i$, which induce two color vectors $\overline{a_1}$ and $\overline{a_2}$, respectively.
We set $D[i, \overline{\mu}] = 1$ if and only if there exists a $j \in \{1,2\}$ such that $D[i - 1, \overline{\mu} - \overline{a_j}] = 1$ holds.
This corresponds to extending a solution $\langle\prec, \sigma_{i - 1}\rangle$ inducing $\overline{\mu} - \overline{a_j}$ to a solution $\langle\prec, \sigma_i\rangle$ inducing $\overline{\mu}$ by using the first (second) page assignment for the edges represented by $C_i$; see also \Cref{fig:stack-layout-dp-transition}.

Evaluating a single entry $D[i, \overline{\mu}]$ takes $\BigO{\chi}$ time.
There are $\BigO{{(m + 1)}^{\chi}}$ different color vectors and we can fill $D$ in $\BigO{\chi \cdot \conComp \cdot (m+1)^{\chi}}$ time.
\instance admits a solution if and only if we can find a color vector $\overline{\mu} = (\mu_1, \ldots, \mu_{\chi})$ with $D[\conComp, \overline{\mu}] = 1$ such that $\mu_c \leq b_1(c)$, and $\Size{\delta^{-1}(c)} - \mu_c \leq b_2(c)$ holds for every color $c \in [\chi]$.
\end{proofsketch}
\begin{prooflater}{plemmaStackColorsXP}
Let $\instance = \instanceLongStack$ be an instance of \capacitatedStackLongTwoPages.
Since $G$ must be planar to admit a 2-page stack layout, we assume that $m \in \BigO{n}$ holds.
As a preprocessing step, we compute the conflict graph $\stackConflict$ in $\BigO{n^2}$ time.
We reject \instance if it is not bipartite, which is a necessary condition.
Otherwise, we consider an arbitrary order $C_1, \ldots, C_{\conComp}$ of its connected components; recall $\conComp = \Size{\mathcal{C}(\stackConflict)}$.

We are now ready to describe our dynamic programming (DP) algorithm.
Let $i \in [\conComp]_0$ be an index and $\overline{\mu} = (\mu_1, \ldots, \mu_{\chi})$ be a \emph{color vector} with $\mu_c \in [m]_0$ for every color $c \in [\chi]$.
Note that there are $\BigO{(m+1)^{\chi}}$ different color vectors.
We will make use of a binary DP-table $D$ where we set $D[i, \overline{\mu}] = 1$ if and only if there exists a solution $\langle\prec, \sigma_i\rangle$ for the edges $\bigcup_{j \in [i]} E(C_j)$ with $\mu_c$ edges of color $c$ on the first page, for every $c \in [\chi]$.
In this case, the solution $\langle\prec, \sigma_i\rangle$ \emph{induces} $\overline{\mu}$.
As our base case, we define $D[0, \overline{\mu}] = 1$ if and only if $\overline{\mu} = \langle 0, \ldots, 0\rangle$.
Clearly, without placing any edge, any solution can only induce the zero vector.
For some $i \in [\conComp]$, observe that the connected component $C_i$ is bipartite.
As a consequence, there are precisely two possible crossing-free page assignments for the edges from $C_i$.
Let $\overline{a_1} = (\mu_1^1, \ldots, \mu_{\chi}^1)$ and $\overline{a_2} = (\mu_1^2, \ldots, \mu_{\chi}^2)$ be their corresponding color vectors for the first page.
Let $\overline{\mu}$ be a color vector.
We set $D[i, \overline{\mu}] = 1$ if and only if there exists a $j \in \{1,2\}$ such that $\overline{\mu} - \overline{a_j}$ has only non-negative entries and $D[i - 1, \overline{\mu} - \overline{a_j}] = 1$ holds.
This is correct since a solution that witnesses $D[i - 1, \overline{\mu} - \overline{a_j}] = 1$ can be extended to a solution for $D[i, \overline{\mu}] = 1$ by using the corresponding page assignment; see also \Cref{fig:stack-layout-dp-transition}.
Moreover, a solution that witnesses $D[i, \overline{\mu}] = 1$ also witnesses $D[i - 1, \overline{\mu} - \overline{a_j}] = 1$ by removing the edges from $C_i$.
Evaluating a single entry $D[i, \overline{\mu}]$ takes $\BigO{\chi}$ time.
Since $D$ has $\BigO{\conComp \cdot (m+1)^{\chi}}$ entries, we can compute $D$ in $\BigO{\chi \cdot \conComp \cdot (m+1)^{\chi}}$ time.

Once the table $D$ has been computed, we need to check whether some color vector fulfills the budget constraints.
To this end, consider a color vector $\overline{\mu} = (\mu_1, \ldots, \mu_{\chi})$.
Observe that a solution $\linL$ that induces $\overline{\mu}$ has $\Size{\delta^{-1}(c)} - \mu_c$ edges of color $c$ on page two, for every $c \in [\chi]$.
Using this observation, we conclude that \instance admits a solution if and only if there exists a color vector $\overline{\mu} = (\mu_1, \ldots, \mu_{\chi})$ such that $D[\conComp, \overline{\mu}] = 1$, $\mu_c \leq b_1(c)$, and $\Size{\delta^{-1}(c)} - \mu_c \leq b_2(c)$ holds for every color $c \in [\chi]$.
We can test this in $\BigO{(m+1)^{\chi}}$ time as a postprocessing step.
Therefore, our DP-algorithm has a running time of $\BigO{\chi \cdot \conComp \cdot (m+1)^{\chi} + n^2}$ and the statement follows.
\end{prooflater}

\subparagraph*{\textsf{W}[1]-Hardness.}
Next, we show that \capacitatedStack is \W[1]-hard when parameterized by the number of colors $\chi$.
To this end, we devise a parameterized reduction from the problem \probname{Multidimensional Subset Sum}, defined as follows.
Given a multiset $S$ of $\eta$ $d$-dimensional vectors over $\mathbb{N}_0^d$ and a $d$-dimensional target vector $\overline{t}$, does there exist a subset $S' \subseteq S$ such that $\sum_{\overline{s} \in S'} \overline{s} = \overline{t}$?
The problem is known to be \W[1]-hard when parameterized by the dimension~$d$, even when the numbers are unary encoded~\cite{GKO.SPB.2020}.

\shortLong{
We represent an integer $x \in \mathbb{N}$ by a \emph{number gadget}, which is a monochromatic star $K_{1,x}$ with center $u_x$ and leaves $v_x^1, \ldots, v_x^x$.
A \emph{vector gadget} represents a $d$-dimensional vector $\overline{s}$ by nesting $d$ number gadgets and coloring them with colors $1, \ldots, d$.
Since we only have two pages, we can use an additional edge $a_{\overline{s}}b_{\overline{s}}$ (in color $d + 1$) to ensure that all edges of a vector gadget must be placed on the same page; see \Cref{fig:stack-layout-colors-vector} for a visualization and \Cref{app:stack-layouts-colors} for the formal construction.
This allows us to make the following crucial observation:
}{
In our reduction, we have overall $d+1$ colors.
Each number is represented as a star and a vector is encoded as a set of $d$ nested stars, one in each color $c \in [d]$.
Since we only have two pages, we can use an additional edge in color $d + 1$ to force all of these edges to be placed on the same page\onlyShort{; see also \Cref{fig:stack-layout-colors-vector}}.
Placing them on the first page represents membership in the solution set $S' \subseteq S$ and by appropriately setting the budget for both pages, we can ensure equivalence between the solutions to both problems.
}

\begin{figure}
	\centering
	\includegraphics[page=1]{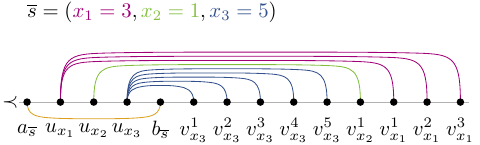}
	\caption{The vector gadget which features the number gadgets for 3, 1, and 5, respectively.}
	\label{fig:stack-layout-colors-vector}
\end{figure}

\begin{statelater}{stackColorsNumberGadget}
    \subparagraph*{Number Gadget.}
    The \emph{number gadget} for an integer $x \in \mathbb{N}$, $x \geq 1$, consists of the vertices $u_x$ and $v_x^i$, $i \in [x]$, arranged in the order $u_x \prec v_x^1 \prec \ldots \prec v_x^{x}$.
    We connect $u_x$ with every $v_x^i$, $i \in [x]$, i.e., introduce a star $K_{1,x}$ with center $u_x$; see also \Cref{fig:stack-layout-colors-vector}.
    The number $x = 0$ is represented by the empty number gadget.
    To represent a vector $\overline{s} = (x_1, \ldots, x_d)$, we combine $d$ number gadgets into a \emph{vector gadget} as follows, where $a_{\overline{s}}b_{\overline{s}}$ is a new edge that we will use to ensure correctness of the reduction; see again \Cref{fig:stack-layout-colors-vector} for an example.
    \begin{align*}
    	a_{\overline{s}} \prec u_{x_1} \prec u_{x_2} \prec \ldots \prec u_{x_d} \prec b_{\overline{s}} \prec v_{x_d}^1 \prec \ldots \prec v_{x_d}^{x_d} \prec v_{x_{d - 1}}^1 \prec \ldots \prec v_{x_1}^{x_1}
    \end{align*}
    We color the edges of the number gadgets based on the dimension they represent, i.e., we have $\delta(u_{x_i}v_{x_i}^j) = i$ for every $i \in [d]$, $j \in [x_i]$.
    Moreover, we color the edge $a_{\overline{s}}b_{\overline{s}}$ in the color ``for'' dimension $d+1$: $\delta(a_{\overline{s}}b_{\overline{s}}) = d+1$.
    Note that the edge $a_{\overline{s}}b_{\overline{s}}$ is in conflict with every other edge from the vector gadget, which leads to the following observation.
\end{statelater}

\begin{figure}[b]
	\centering
	\includegraphics[page=1]{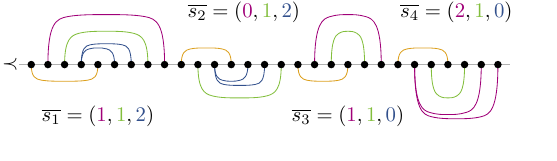}
	\caption{An instance of \capacitatedStack constructed for the instance $(S = \{\overline{s_1}, \overline{s_2}, \overline{s_3}, \overline{s_4}\}, \overline{t} = (2,2,2))$ of \probname{Multidimensional Subset Sum}. The layout depicts the solution $S' = \{\overline{s_1}, \overline{s_3}\}$.}
	\label{fig:stack-layout-colors-example-reduction}
\end{figure}
\begin{observation}
	\label{obs:stack-layouts-colors}
	In every 2-page stack layout \linL of the vector gadget for a vector $\overline{s} = (x_1, \ldots, x_d)$ we have $\sigma(a_{\overline{s}}b_{\overline{s}}) \neq \sigma(u_{x_i}v_{x_i}^j)$ for every $i \in [d]$ and $j \in [x_i]$.
	In particular, we have $\sigma(u_{x_y}v_{x_y}^i) = \sigma(u_{x_{z}}v_{x_{z}}^{j})$ for every $y, z \in [d]$, $i \in [x_y]$, and $j \in [x_{z}]$.
\end{observation}
\shortLong{
We add a vector $\overline{s}$ to the solution set $S'$ of an instance $(S, \overline{t} = (t_1, \ldots, t_d))$ of \probname{Multidimensional Subset Sum} if and only if the edge $a_{\overline{s}}b_{\overline{s}}$ is assigned to the second page; see also \Cref{fig:stack-layout-colors-example-reduction} for a small example.
To ensure correctness of the reduction, we set $b_1(c) = t_c$ and $b_2(c) = \Size{\delta^{-1}(c)} - t_c$ for every $c \in [d]$.
Finally, we set $b_1(d + 1) = b_2(d + 1) = \Size{S}$.
Observe that due to their budget $b_p(c)$ both pages $p = 1,2$ must contain $b_p(c)$ edges of color $c$ in a solution, for every color $c \in [d]$.
Together with \Cref{obs:stack-layouts-colors}, this allows us to show:%
}{
We now use \Cref{obs:stack-layouts-colors} to establish correctness of our reduction; see also \Cref{fig:stack-layout-colors-example-reduction} for a small example of the construction.
}%

\begin{restatable}\restateref{lem:stack-layouts-colors-w-one}{lemma}{lemmaStackColorsWOne}
	\label{lem:stack-layouts-colors-w-one}
	\capacitatedStackLongTwoPages is \W\textup{[1]}-hard when parameterized by the number of colors $\chi$.
\end{restatable}
\begin{prooflater}{plemmaStackColorsWOne}
	Let $(S, \overline{t})$ be an instance of \probname{Multidimensional Subset Sum} on $d$ dimensions with $\overline{t} = (t_1, \ldots, t_d)$ such that every number can be bounded by a polynomial in $\eta = \Size{S}$; we recall that \probname{Multidimensional Subset Sum} remains \W[1]-hard parameterized by $d$ in this case~\cite{GKO.SPB.2020}.
	Consider an arbitrary order of the $\eta$ vectors in $S = \{\overline{s_1}, \ldots, \overline{s_{\eta}}\}$.
	We now create an instance of \capacitatedStack by first creating the vector gadget for every vector $\overline{s_i}$, $i \in [\eta]$ and placing it on the spine $\prec$ ordered by index.
	For every color $c \in [d]$, we set $b_1(c) = t_c$ and $b_2(c) = (\sum_{(\ldots, x_c, \ldots) \in S} x_c) - t_c = \Size{\delta^{-1}(c)} - t_c$.
	Finally, we set $b_1(d + 1) = b_2(d + 1) = \eta$.
	Let $\instance = \instanceLongStack$ be the obtained instance.
	We observe that its size is linear in $\eta + \sum_{\overline{s} \in S}\sum_{x_i \in \overline{s}} x_i$, which is polynomial in $\eta$.
	Moreover, we have $\chi = d + 1$.
	It remains to show correctness of the reduction to obtain \W[1]-hardness by $\chi$.
	
	\proofsubparagraph{($\boldsymbol{\Rightarrow}$)}
	Assume that $(S, \overline{t})$ is a positive instance of \probname{Multidimensional Subset Sum} and let $S' \subseteq S$ be a set of vectors such that $\sum_{\overline{s} \in S'} \overline{s} = \overline{t}$.
	We now construct a solution \linL for \instance as follows.
	For every $\overline{s_i} \in S$, $i \in [\eta]$, we set $\sigma(a_{\overline{s_i}}b_{\overline{s_i}}) = 1$ if and only if $\overline{s_i} \notin S'$.
	We assign all other edges of the corresponding vector gadget to the respective other page.
	Observe that no two edges of the vector gadget cross in this page assignment; note that $\sigma$ resembles a necessary page assignment established in \Cref{obs:stack-layouts-colors}.
	Moreover, the edges of the number gadgets are placed on the first page if and only if $\overline{s_i}$ is part of the solution, i.e., $\overline{s_i} \in S'$.
	Since all gadgets are disjoint with respect to $\prec$, the obtained linear layout \linL is a stack layout.
	To see that it is also admissible, we observe, on the one hand, that for every color $c \in [d]$ we have exactly $t_c = b_1(c)$ edges of color $c$ on the first page.
	This is because $S'$ is a solution to $(S, \overline{t})$ and thus we have $\sum_{(\ldots, x_c, \ldots) \in S'} x_c = t_c$.
	Consequently, the second page contains $(\sum_{(\ldots, x_c, \ldots) \in S} x_c) - t_c = b_2(c)$ edges of color $c$.
	Regarding color $d + 1$, since $b_1(d + 1) = b_2(d + 1) = \Size{S}$ and we have $\Size{S'}, \Size{S} - \Size{S'} \leq \Size{S}$, the budget for color $d + 1$ is fulfilled.
	Thus, \linL witnesses that \instance is a positive instance of \capacitatedStack.
	
	\proofsubparagraph{($\boldsymbol{\Leftarrow}$)}
	For the reverse direction, assume that \instance is a positive instance of \capacitatedStack and let \linL be a solution.
	To construct a solution $S'$ to $(S, \overline{t})$, we consider every vector $\overline{s_i} = (x_1, \ldots, x_d) \in S$.
	We add $\overline{s_i}$ to $S'$ if and only if $\sigma(a_{\overline{s_i}}b_{\overline{s_i}}) = 2$.
	To see that $S'$ is a solution, recall that, for every $\overline{s_i} \in S$, we have $\sigma(u_{x_y}v_{x_y}^j) = \sigma(u_{x_{z}}v_{x_{z}}^{j'})$ for every $y, z \in [d]$, $j \in [x_y]$, and $j' \in [x_{z}]$ by \Cref{obs:stack-layouts-colors}, i.e., all edges of the number gadgets for $\overline{s_i}$ are assigned to the same page.
	In particular, they are placed on the first page if and only if we have added $\overline{s_i}$ to the solution set $S'$.
	Consider a color $c \in [d]$.
	Since \linL is a solution to \instance, the first page contains at most $b_1(c) = t_c$ edges of color $c$ and the second page contains at most $b_2(c) = \Size{\delta^{-1}(c)} - t_c$ edges of color $c$.
	Moreover, $G$ contains precisely $\sum_{(\ldots, x_c, \ldots) \in S} x_c = b_1(c) + b_2(c)$ edges of color $c$.
	Thus, by applying these arguments for every color $c \in [d]$, we can conclude $\sum_{\overline{s} \in S'} \overline{s} = \overline{t}$, i.e., $(S, \overline{t})$ is a positive instance of \probname{Multidimensional Subset Sum}.
\end{prooflater}

The following theorem summarizes \Cref{lem:stack-layouts-colors-xp,lem:stack-layouts-colors-w-one}.
\begin{theorem}
	\label{thm:stack-layouts-colors}
	\capacitatedStackLongTwoPages is in \XP\ and \W\textup{[1]}-hard when parameterized by the number of colors $\chi$.
\end{theorem}

\subsection{An \textsf{FPT}-Algorithm for Few Colors and Bounded-Size Components}
\label{sec:stack-layouts-colors-components-size}
In this section, we establish our last result for \capacitatedStack, namely a fixed-parameter tractable algorithm when parameterized by the number of colors $\chi$ and the size of the largest connected component in the conflict graph $\conCompSize = \max \{\Size{V(C)} \mid C \in \mathcal{C}(\stackConflict)\}$.
Note that both parameters are necessary for an \FPT-algorithm since we have shown that the problem is \W[1]-hard parameterized by~$\chi$ (\Cref{thm:stack-layouts-colors}) and para\NP-hard by $\conCompSize$ (\Cref{thm:stack-layouts-components-size}).
In the following, we establish a more general result in which $\ell$ is not assumed to be $2$ but a parameter.
While this is not necessary to establish the result of this section, it will later allow us to show that \capacitatedQueue is \FPT\ with respect to $\chi$, $\conCompSize$, \emph{and} the number of pages $\ell$.

Let $\instance = \instanceLongQueue$ be an instance of \capacitatedStackGen.
On a high level, since both $\chi$ and $\conCompSize$ are bounded, we can assign to each connected component of $\stackConflict$ one of $g(\chi, \conCompSize)$-many types, for some computable function $g$.
Moreover, since also the number of pages is bounded, each type has at most $h(\ell, \conCompSize)$ different valid, i.e., crossing-free, page assignments, for some computable function $h$.
We can observe that from the perspective of a solution, all that matters is how many components from each type use a particular page assignment.
The precise distribution of components to page assignments is irrelevant.
Thus, we can encode the problem of finding an admissible layout as an integer linear program (ILP) with a bounded number of variables -- essentially one for each type and their possible page assignment.
Fixed-parameter tractability then follows by checking feasibility for the ILP in \FPT-time:
\begin{proposition}[See, e.g., \cite{Len.IPF.1983,RR.SFC.2023}]
	\label{prop:ilp-fpt}
	Let $\mathcal{P}$ be an integer linear program in $\eta$ variables.
	Then, it can be solved in time $f(\eta) \cdot \Size{\mathcal{P}}^{\BigO{1}}$ for some computable function $f$.
\end{proposition}
\onlyLong{
In the following, we formalize each of the above-hinted steps.}

\subparagraph*{Component Types and Their Embeddings.}
In the following, for an edge $e \in E$ and connected component $C \in \mathcal{C}(\stackConflict)$, we write $e \in V(C)$ to denote that $e$ is represented by a vertex of $C$.
We start by defining the component types.

\begin{definition}
    \label{def:stack-layouts-colors-components-size-types}
    Two connected components $C_i, C_j \in \mathcal{C}(\stackConflict)$ have the same \emph{component type}, denoted $C_i \sim C_j$, if and only if there exists an isomorphism $\xi\colon V(C_i) \to V(C_j)$ from $C_i$ to~$C_j$ such that $\delta(e) = \delta(\xi(e))$ for every $e \in V(C_i)$.
\end{definition}

Note that $\sim$ is an equivalence relation whose number of classes is bounded:

\begin{restatable}\restateref{lem:stack-layouts-colors-components-size-types}{lemma}{lemmaStackColorsComponentsTypes}
    \label{lem:stack-layouts-colors-components-size-types}
    There are $\BigO{\conCompSize \cdot 2^{{\conCompSize}^2} \cdot \chi^{\conCompSize}}$ equivalence classes in $\mathcal{C}(\stackConflict)/{\sim}$.
    We can compute the partition of $\mathcal{C}(\stackConflict)$ into $\sim$ in $\BigO{\conCompSize^3 \cdot 2^{{\conCompSize}^2 + \conCompSize \cdot \log \conCompSize} \cdot \chi^{\conCompSize} \cdot  m}$ time.
\end{restatable}
\begin{prooflater}{plemmaStackColorsComponentsTypes}
    Each connected component $C \in \mathcal{C}(\stackConflict)$ has at most $\conCompSize$ vertices.
    Moreover, every vertex $e \in V(C)$ has one of $\chi$ colors.
    For every $i \in [\conCompSize]$, there are $2^{\binom{i}{2}} \cdot \chi^{i}$ different vertex-colored graphs (on the vertex set $[i]$).
    Thus, for up to $\conCompSize$ vertices, there are 
    \begin{align*}
        \sum_{i = 1}^{\conCompSize} 2^{\binom{i}{2}} \cdot {\chi}^i \leq \conCompSize \cdot 2^{\binom{\conCompSize}{2}} \cdot \chi^{\conCompSize} \in \BigO{\conCompSize \cdot 2^{{\conCompSize}^2} \cdot \chi^{\conCompSize}}
    \end{align*}
    different vertex-colored graphs.
    Every equivalence class in $\mathcal{C}(\stackConflict)/{\sim}$ contains a connected component which is isomorphic to one of these graphs.
    Consequently, we can upper bound the number of equivalence classes in $\mathcal{C}(\stackConflict)/{\sim}$ by $\BigO{\conCompSize \cdot 2^{{\conCompSize}^2} \cdot \chi^{\conCompSize}}$.

    For the running time, we first observe that $\stackConflict$ contains $m$ vertices, one for each edge in $G$, and that the number of connected components is, therefore, also at most $m$.
    For each connected component $C \in \mathcal{C}(\stackConflict)$, we can iterate over all equivalence classes (found so far).
    For each potential equivalence class, we can branch over all $\BigO{\conCompSize!} \in \BigO{2^{\conCompSize\cdot\log\conCompSize}}$ possible bijections $\xi$ and check in $\BigO{{\conCompSize}^2}$ time whether one of them is a color-preserving isomorphism or whether $C$ induces a ``new'' equivalence class in $\mathcal{C}(\stackConflict)/{\sim}$.
    Combining all, the running time bound follows readily.
\end{prooflater}

Since every connected component $C \in \mathcal{C}(\stackConflict)$ has at most $\conCompSize$ vertices, there are at most $\ell^{\conCompSize}$ possible page assignments $\sigma$ for the edges $e \in V(C)$, i.e., vertex-colorings of $C$.
However, we must exclude those that introduce crossings on a page, i.e., only allow proper colorings of $C$:

\begin{definition}
    \label{def:stack-layouts-colors-components-size-embeddings}
    Let $[C]$ be an equivalence class of $\sim$.
    The \emph{valid page assignments} $\Sigma([C])$ is the set of all page assignments $\sigma\colon V(C) \to [\ell]$ such that there are no two edges $ab, cd \in C$ with $\sigma(ab) = \sigma(cd)$ such that $a \prec c \prec b\prec d$.
\end{definition}

We can efficiently compute the valid page assignments $\Sigma([C])$ of an equivalence class $[C]$:
\begin{observation}
    \label{obs:stack-layouts-colors-components-size-embeddings-number-time}
    Let $[C]$ be an equivalence class of $\sim$.
    It holds $\Size{\Sigma([C])} \in \BigO{\ell^{\conCompSize}}$.
    We can compute the set $\Sigma([C])$ in $\BigO{\ell^{\conCompSize} \cdot \conCompSize^2}$ time.
\end{observation}

For an equivalence class $[C]$ of $\sim$, a valid page assignment $\sigma \in \Sigma([C])$, and a color $c \in [\chi]$, we let $\#_p([C], \sigma, c)$ denote the number of edges of color $c$ assigned to page $p \in [\ell]$ in the page assignment $\sigma$.
We are now ready to state our ILP.

\subparagraph*{Using an ILP to Find a Solution.}
We now define the ILP $\mathcal{P}(\instance)$ for \instance.
For every equivalence class $[C]$ of $\sim$, we define $n_{[C]} \coloneqq \Size{[C]} = \Size{\{C' \in \mathcal{C}(\stackConflict) \mid C' \sim C\}}$ to be the number of connected components $C'$ of $\stackConflict$ of component type $[C]$.
For every equivalence class $[C]$ of $\sim$ and every page assignment $\sigma \in \Sigma([C])$, we introduce an integer variable $x_{[C],\sigma}$ that captures the number of connected components $C' \in [C]$ %
whose vertices, i.e., edges of $G$, are assigned to pages according to $\sigma$ in a solution.
Observe that every variable assignment with $0 \leq x_{[C],\sigma} \leq n_{[C]}$ and $\sum_{\sigma \in \Sigma([C])} x_{[C],\sigma} = n_{[C]}$ for every $[C] \in \mathcal{C}(\stackConflict)/{\sim}$ and $\sigma \in \Sigma([C])$ can be translated into a crossing-free linear layout of $G$.
To capture this, we introduce the following constraints:
\begin{align}
    0 \leq x_{[C],\sigma} \leq n_{[C]}&\quad\quad \forall [C] \in \mathcal{C}(\stackConflict)/{\sim},\ \sigma\in\Sigma([C])\tag{BND}\label{eq:stack-layouts-colors-components-size-bound}\\
    \sum_{\sigma \in \Sigma([C])} x_{[C],\sigma} = n_{[C]}&\quad\quad \forall [C] \in \mathcal{C}(\stackConflict)/{\sim}\tag{SUM}\label{eq:stack-layouts-colors-components-size-sum}
\end{align}
We need additional constraints to ensure that a layout also respects our budget constraints.
To this end, recall that we let $\#_p([C], \sigma, c)$ denote the number of edges $e \in V(C)$ of color $c \in [\chi]$ assigned to page $p \in [\ell]$ in the page assignment $\sigma \in \Sigma([C])$.
Note that this value is a constant that can be precomputed.
To get the overall number of edges of color $c$ on some page $p$, we can sum this value up over all valid page assignments for all equivalence classes $[C] \in \mathcal{C}(\stackConflict)/{\sim}$,
which should not exceed the budget constraint for color $c$ on page $p$:

\begin{align}
    \sum_{[C] \in \mathcal{C}(\stackConflict)/{\sim}}\ \sum_{\sigma\in\Sigma([C])}  \#_p([C], \sigma, c) \cdot x_{[C],\sigma} \leq b_p(c)\quad\quad  
    \forall c \in [\chi],\ p \in [\ell]\tag{BGT}\label{eq:stack-layouts-colors-components-size-budget}
\end{align}

\onlyLong{
This concludes the definition of the ILP.
Next, we show that it is feasible if and only if our instance of \capacitatedStack admits a solution.
We use this result to complete our \FPT-algorithm for \capacitatedStack parameterized by $\ell+\conCompSize + \chi$.
}

\begin{restatable}\restateref{thm:stack-layouts-pages-colors-components-size}{theorem}{theoremStackColorsComponents}
	\label{thm:stack-layouts-pages-colors-components-size}
    \capacitatedStackLong is \FPT\ in the number of pages $\ell$, number of colors $\chi$, and the maximum size $\conCompSize$ of a connected component of $\stackConflict[G, \prec]$.
\end{restatable}
\begin{prooflater}{ptheoremStackColorsComponents}
    Let $\instance = \instanceLongQueue$ be an instance of \capacitatedStackLong.
    We compute the conflict graph $\stackConflict$ and the partition of its connected components into the equivalence classes $\mathcal{C}(\stackConflict)/{\sim}$ in $\BigO{m^2 + \conCompSize^3 \cdot 2^{{\conCompSize}^2 + \conCompSize \cdot \log \conCompSize} \cdot \chi^{\conCompSize} \cdot m}$ time; recall also \Cref{lem:stack-layouts-colors-components-size-types}.
    Next, we construct the above-introduced integer linear program $\mathcal{P}(\instance)$.
    Observe that $\mathcal{P}(\instance)$ has one variable for every valid page assignment of every equivalence class of $\sim$.
    Thus, by \Cref{lem:stack-layouts-colors-components-size-types,obs:stack-layouts-colors-components-size-embeddings-number-time}, we can bound the number of variables in $\mathcal{P}(\instance)$ by $\BigO{\conCompSize \cdot 2^{{\conCompSize}^2} \cdot \chi^{\conCompSize} \cdot \ell^{\conCompSize}}$.
    We can construct $\mathcal{P}(\instance)$ in $\BigO{\conCompSize \cdot 2^{{\conCompSize}^2} \cdot \chi^{\conCompSize} \cdot \ell^{\conCompSize}\cdot \conCompSize^2 + \ell\cdot \chi}$ time.
    By \Cref{prop:ilp-fpt}, we can check in $g(\conCompSize \cdot 2^{{\conCompSize}^2} \cdot \chi^{\conCompSize} \cdot \ell^{\conCompSize}) \cdot \Size{\mathcal{P}(\instance)}^{\BigO{1}}$ time whether $\mathcal{P}(\instance)$ admits a solution for some computable function $g$.
    Since the size of $\mathcal{P}(\instance)$ can be bounded by $h(\ell, \conCompSize, \chi) \cdot n^{\BigO{1}}$ for some computable function $h$, we can bound the overall running time by $f(\ell, \conCompSize, \chi) \cdot n^{\BigO{1}}$ for some computable function $f$.
    To establish the statement, it remains to connect solutions for $\mathcal{P}(\instance)$ with those for \instance.
    
    \begin{claim}
    	\label{clm:stack-layouts-colors-components-size}
    	\instance is a positive instance if and only if  $\Sigma([C]) \neq \emptyset$ for every equivalence class $[C] \in \mathcal{C}(\stackConflict)/{\sim}$ and $\mathcal{P}(\instance)$ is feasible.
    \end{claim}
    \begin{claimproof}
    	For the sake of the proof, assume that $\Sigma([C]) \neq \emptyset$ holds for every equivalence class $[C] \in \mathcal{C}(\stackConflict)/{\sim}$.
    	Otherwise, the equivalence is trivial.
    	We now argue the remaining part of the if and only if.
    	
    	\proofsubparagraph{($\boldsymbol{\Rightarrow}$)}
    	Assume that $\instance$ is a positive instance and consider a solution $\linL$.
    	Since $\linL$ is a solution, it is in particular a valid stack layout.
    	Thus, every connected component $C \in\mathcal{C}(\stackConflict)$ uses %
        one of its valid page assignments $\sigma' \in \Sigma([C])$.
    	We now consider every equivalence class $[C]$ of $\sim$ and every valid page assignment $\sigma' \in \Sigma([C])$.
    	In line with the intended semantics of our variables, we set $x_{[C],\sigma'}$ equal to the number of connected components from $\stackConflict \cap [C]$ that use the page assignment $\sigma'$.
    	Since these are at least zero and at most $n_{[C]}$, this variable assignment fulfills Constraint~\eqref{eq:stack-layouts-colors-components-size-bound}.
        Additionally, since every connected component $C' \in [C]$ uses precisely one of its valid page assignments, the sum of these variables must equal $n_{[C]}$, i.e., we also fulfill Constraint~\eqref{eq:stack-layouts-colors-components-size-sum}.
    	Moreover, since $\linL$ is an admissible layout, the number of edges of color $c$ on every page $p \in [\ell]$ is at most $b_p(c)$.
    	Thus, the variable assignment also satisfies Constraint~\eqref{eq:stack-layouts-colors-components-size-budget} and hence witnesses that $\mathcal{P}(\instance)$ is feasible.
    	
    	\proofsubparagraph{($\boldsymbol{\Leftarrow}$)}
    	Assume that $\mathcal{P}(\instance)$ is feasible (and that $\Sigma([C]) \neq \emptyset$ holds for every equivalence class $[C] \in \mathcal{C}(\stackConflict)/{\sim}$) and consider a variable assignment that satisfies Constraints~\eqref{eq:stack-layouts-colors-components-size-bound},~\eqref{eq:stack-layouts-colors-components-size-sum}, and~\eqref{eq:stack-layouts-colors-components-size-budget}.
    	We now construct a solution $\linL$ as follows.
    	Consider an equivalence class $[C]$ of $\mathcal{C}(\stackConflict)/{\sim}$ and all the connected components $C_1, \ldots, C_{n_{[C]}} \in \stackConflict \cap [C]$ of the same component type (ordered according to their occurrence along $\prec$).
        Moreover, consider an arbitrary order of the valid page assignments $\sigma_1, \ldots, \sigma_k$ of $\Sigma([C])$; note that $k \geq 1$ holds by assumption.
        For the first $x_{[C], \sigma_1}$ connected components $C_1, \ldots, C_{x_{[C], \sigma_1}}$, we set $\sigma(e) = \sigma_1(\xi_i(e))$ (for the color-preserving isomorphism $\xi_i\colon C_i \to C$) for every edge $e \in V(C_i)$, $1 \leq i \leq x_{[C],\sigma_1}$.
        For the next $x_{[C], \sigma_2}$ connected components, we set $\sigma(e) = \sigma_2(\xi_i(e))$ for every edge $e \in V(C_i)$, $x_{[C], \sigma_1} + 1 \leq i \leq x_{[C], \sigma_1} + x_{[C], \sigma_2}$.
        We continue this process until we have considered all $k$ valid page assignments of $[C]$.
    	Note that this assignment is well-defined as we satisfy Constraints~\eqref{eq:stack-layouts-colors-components-size-bound} and~\eqref{eq:stack-layouts-colors-components-size-sum}.
    	We repeat this step for every equivalence class $[C]$ of $\mathcal{C}(\stackConflict)/{\sim}$ to eventually obtain $\linL$.
    	To see that $\linL$ is a solution, observe, on the one hand, that it is a valid $\ell$-page stack layout since %
        we only use valid page assignments and every equivalence class $[C] \in \mathcal{C}(\stackConflict)/{\sim}$ admits at least one such page assignment, i.e., $\Sigma([C]) \neq \emptyset$.
    	Moreover, two edges from different connected components cannot cross by the definition of connected components in $\stackConflict$.
    	On the other hand, $\linL$ is admissible as the variable assignment satisfies Constraint~\eqref{eq:stack-layouts-colors-components-size-budget}.
    	Thus, $\linL$ is a solution.
    \end{claimproof}
    
    The theorem follows by combining \Cref{clm:stack-layouts-colors-components-size} with the running time bound.
\end{prooflater}
\onlyLong{\Cref{thm:stack-layouts-pages-colors-components-size} immediately implies:}
\begin{corollary}
	\label{cor:stack-layouts-colors-components-size}
    \capacitatedStackLongTwoPages is \FPT\ in the number of colors $\chi$ and the maximum size $\conCompSize$ of a connected component of $\stackConflict$.
\end{corollary}

\section{Fixed Order Queue Layouts}
\label{sec:queue-layouts}
In this section, we investigate queue layouts with fixed spine order $\prec$.
It is well-known that an $\ell$-page queue layout with spine order $\prec$ exists if and only if $\prec$ does not induce an $(\ell + 1)$-rainbow, i.e., a set of $\ell + 1$ pairwise nesting edges~\cite{HR.LOG.1992}.
Heath and Rosenberg~\cite{HR.LOG.1992} provided an efficient algorithm to construct such a queue layout (if it exists).
However, in their algorithm they employ a greedy-like strategy and assign the edges to the first page in which they do not nest any other edge.
The presence of our budgets invalidates such a strategy, and we show that deciding existence of a solution is already \NP-complete for monochromatic graphs, i.e., where all edges have the same color, when $\ell$ is part of the input.
Before we establish said result, we analyze the complexity of our problem %
when $\ell \in \{2,3\}$.

\subsection{Instances with Two or Three Pages}
\label{sec:queue-layouts-two-pages}
We first consider \capacitatedQueue when restricted to instances with two pages, i.e., $\ell = 2$.
In particular, we show that the problem of finding an admissible 2-page queue layout is \NP-complete, even if every connected component in $\queueConflict$ has constant size.
Afterwards, we present with \Cref{thm:queue-layouts-two-pages-components-fpt} an \FPT-algorithm parameterized by the number of connected components of the conflict graph \queueConflict\ for the case $\ell = 2$.
Finally, we show in \Cref{thm:queue-layouts-three-pages-connected} that the aforementioned \FPT-algorithm cannot be lifted to instances with $\ell = 3$.

\subparagraph*{Hardness for Two Pages.}
We start by establishing hardness for two pages even if every connected component in $\queueConflict$ has at most five vertices.
Observe that \Cref{thm:stack-layouts-components-size} is the counterpart to this result for \capacitatedStack.
Moreover, the overall proof idea and the gadgets that we employ are very similar; the main part that differs between the two reductions is the linear order $\prec$.
\shortLong{In the interest of space, we only provide an illustration of the gadgets (see \Cref{fig:queue-layout-component-size-gadgets}) but defer all formal arguments to establish the following result to \Cref{app:queue-layouts-two-pages}:}{For the sake of completeness, we provide in the following all formal arguments to establish the claimed hardness but focus on the part that differs between the two reductions.}

\begin{figure}
	\centering
	\includegraphics[page=1]{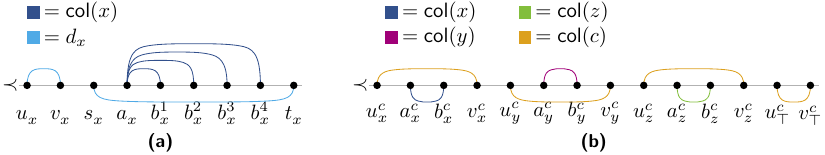}
	\caption{Gadgets representing \textbf{\textsf{(a)}} the variable $x$ in its true state and \textbf{\textsf{(b)}} the clause $c = (x \lor y \lor z)$ where $x$ and $z$ are set to true and $y$ to false in the reduction for \Cref{thm:queue-layouts-components-size}.}
	\label{fig:queue-layout-component-size-gadgets}
\end{figure}
\begin{statelater}{queueLayoutsComponentsSizeHardnessDetails}
For the following discussion, let $\varphi$ be an instance of the \NP-complete problem \probname{Monotone Not All Equal 3Sat-E4}~\cite{DD.shv.2020}; see again \Cref{sec:stack-layouts-components-size} for its definition.

\subparagraph*{Variable Gadget.}
Let $x \in X$ be a variable of $\varphi$.
As in the construction for \Cref{thm:stack-layouts-components-size}, the variable gadget for $x$ consists of a selecting edge $s_xt_x$, a dummy edge $u_xv_x$, and four variable edges $a_xb_x^i$, $i = 1,2,3,4$.
The variable edges receive the color $\col{x}$, the remaining two edges receive the color $d_x \neq \col{x}$.
While the budget for these colors is the same as in \Cref{sec:stack-layouts-components-size}, i.e., $b_1(d_x) = b_2(d_x) = 1$ and $b_1(\col{x}) = b_2(\col{x}) = 4$, the order of the edges' endpoints differs.
More concretely, we use the order $u_x \prec v_x \prec s_x \prec a_x \prec b_x^1 \prec b_x^2 \prec b_x^3 \prec b_x^4 \prec t_x$; see also \Cref{fig:queue-layout-component-size-gadgets}a.
Due to the budget and the placement of the endpoints (note that the selecting edge nests all variable edges), we can again make the following observation.
\begin{observation}
	\label{obs:queue-layouts-components-size-variable}
	In every admissible 2-page queue layout \linL of the variable gadget for a variable $x \in X$ we have $\sigma(s_xt_x) \neq \sigma(u_xv_x) = \sigma(a_xb_x^i)$ for every $i = 1,2,3,4$.
\end{observation}

\subparagraph*{Clause Gadget.}
Let $c = (x \lor y \lor z) \in C$ be a clause of $\varphi$ for some $x,y,z \in X$; recall that $\varphi$ is monotone, i.e., every clause contains only positive literals.
Similar to the variable gadget, we also use for the clause gadget the same set of colored edges as in \Cref{sec:stack-layouts-components-size} but place their endpoints differently on the spine $\prec$.
More formally, the clause gadget has three literal edges $a_{x}^cb_{x}^c$, $a_{y}^cb_{y}^c$, and~$a_z^cb_{z}^c$, each of which inherits the color of the respective variable gadget, i.e., $\delta({a_{q}^cb_{q}^c}) = \col{q}$ for $q \in \{x, y, z\}$.
Moreover, the gadget has four validator edges $u_x^cv_x^c$, $u_{y}^cv_{y}^c$, $u_{z}^cv_{z}^c$, and $u_{\top}^cv_{\top}^c$ which have the color $\col{c}$.
We set the budget $b_1(\col{c}) = b_2(\col{c}) = 2$ and ensure that each literal edge is nested precisely by the corresponding validator edge, i.e., $u_{x}^c \prec a_{x}^c \prec b_{x}^c \prec v_{x}^c \prec u_{y}^c \prec a_{y}^c \prec b_{y}^c \prec v_{y}^c \prec u_{z}^c \prec a_{z}^c \prec b_{z}^c \prec v_{z}^c \prec u_{\top}^c \prec v_{\top}^c$; see also \Cref{fig:queue-layout-component-size-gadgets}b.
We now show the following analogue to \Cref{lem:stack-layouts-components-size-clause}:

\begin{lemma}
	\label{lem:queue-layouts-components-size-clause}
	Let \linL be an admissible 2-page queue layout of the constructed instance that contains the clause gadget for a clause $c \in C$ with $c = (x \lor y \lor z)$.
	We have $\sigma(a_{q}^cb_{q}^c) = \sigma(s_qt_q)$ for every $q \in \{x,y,z\}$.
	Moreover, there exist two literals $p,q \in \{x,y,z\}$ such that $\sigma(a_{p}^cb_{p}^c) \neq \sigma(a_{q}^cb_{q}^c)$.
\end{lemma}
\begin{proof}
	Let $c \in C$ be a clause of $\varphi$ with $c = (x \lor y \lor z)$.
    The arguments for the first statement are analogous to those for \Cref{lem:stack-layouts-components-size-clause}.

    For the second statement, we observe that $u_{q}^c \prec a_{q}^c \prec b_{q}^c \prec v_{q}^c$ implies $\sigma(a_{q}^cb_{q}^c) \neq \sigma(u_{q}^cv_{q}^c)$ for every $q \in \{x,y,z\}$.
    Thus, any solution \linL with $\sigma(a_{x}^cb_{x}^c) = \sigma(a_{y}^cb_{y}^c) = \sigma(a_{z}^cb_{z}^c)$ must also have $\sigma(u_{x}^cv_{x}^c) = \sigma(u_{y}^cv_{y}^c) = \sigma(u_{z}^cv_{z}^c)$.
	Since these edges have color $\col{c}$ we derive that \linL cannot be a solution since $b_1(\col{c}) = b_2(\col{c}) = 2 < 3$.
    Hence, if \linL is a solution, then $\sigma(a_{x}^cb_{x}^c) = \sigma(a_{y}^cb_{y}^c) = \sigma(a_{z}^cb_{z}^c)$ cannot hold, and we can find $p,q \in c$ that fulfills the second claim.
\end{proof}

\subparagraph*{Putting Everything Together.}
For an instance $\varphi$ of \probname{Monotone Not All Equal 3Sat-E4} with variables $X = \{x_1, \ldots, x_\eta\}$ and clauses $C = \{c_1, \ldots, c_\mu\}$, we replace every variable and clause with its corresponding gadget and order them on the spine by their index (where variables precede the clauses).
The following theorem establishes correctness of the reduction.
\end{statelater}

\begin{restatable}\restateref{thm:queue-layouts-components-size}{theorem}{theoremQueueComponentsSize}
	\label{thm:queue-layouts-components-size}
	\capacitatedQueueLong is \NP-complete even when restricted to balanced and tight instances $\instanceLongQueue$ 
    with $\ell = 2$ and where every connected component of $\queueConflict(G, \prec)$ has at most five vertices.
\end{restatable}
\begin{prooflater}{ptheoremQueueComponentsSize}
	Containment in \NP\ is immediate and we proceed with showing \NP-hardness.
    Let $\varphi$ be an instance of \probname{Monotone Not All Equal 3Sat-E4} on $\eta$ variables $X$ and $\mu$ clauses $C$.
    We construct the above-described instance $\instance = \instanceLongQueueTwoPages$ of \capacitatedQueue.
    A careful analysis of the gadgets shows that it can be constructed in $\BigO{\eta + \mu}$ time, has two pages, $\BigO{\eta + \mu}$ colors, and for each color $c \in [\chi]$ it holds $b_1(c) = b_2(c) = \Size{\delta^{-1}(c)} / 2$.
    For the same reason as for \Cref{thm:stack-layouts-components-size}, we can observe that the largest connected component of $\queueConflict$ is induced by the variable gadgets and consists of five vertices; recall also \Cref{fig:queue-layout-component-size-gadgets}.
    We conclude the proof by establishing correctness of the reduction.
	
	\proofsubparagraph{($\boldsymbol{\Rightarrow}$)}
	Assume that $\varphi$ admits a solution $\Psi\colon X \to \{0,1\}$.
    We now show that also \instance admits a solution \linL.
	To this end, consider a variable $x \in X$ and its corresponding variable gadget.
    As in the proof of \Cref{thm:stack-layouts-components-size}, we assign $\sigma(s_xt_x)$ to the first page if and only if $x$ is set to false.
	All other edges are assigned to the page $p \in \{1,2\} \setminus \{\sigma(s_xt_x)\}$; recall also \Cref{obs:queue-layouts-components-size-variable}.
    Since edges in the variable gadget are only in conflict with the selecting edge (if at all), this assignment does not introduce any nesting.
    Since the selecting and dummy edge are on different pages (and the only edges of color $d_x$), we match the budget for $d_x$ on both pages.
	For a clause $c = (x \lor y \lor z) \in C$, we assign the literal edge $a_{q}^cb_{q}^c$ to the page $\sigma(s_qt_q)$ for every $q \in \{x,y,z\}$.
    To prevent nesting, we set $\sigma(u_{q}^cv_{q}^c) \neq \sigma(a_{q}^cb_{q}^c)$.
    Since $\Psi$ witnesses the satisfiability of $\varphi$, there exists a unique variable $q \in \{x,y,z\}$ whose truth assignment differs from the other two.
	We assign $u_{\top}^cv_{\top}^c$ to the same page as $u_{q}^cv_{q}^c$, i.e., $\sigma(u_{\top}^cv_{\top}^c) = \sigma(u_{q}^cv_{q}^c)$.
    As $u_{\top}^cv_{\top}^c$ is an isolated vertex in $\queueConflict$, this does not introduce any nesting.
    Moreover, each page contains precisely two edges of color $\col{c}$ and we also have $b_1(\col{c}) = b_2(\col{c}) = 2$.
    To finally conclude that \linL is a solution to \instance, we make the observation that in the variable gadget for any variable $x \in X$, all four edges $a_xb_x^i$, $i = 1,2,3,4$, are assigned to the same page, which is different from $\sigma(s_xt_x)$.
    For each of the four clauses $c \in C$ that contain the variable $x$, we have $\sigma(s_xt_x) = \sigma(a_{x}^cb_{x}^c)$.
    Thus, every page contains exactly four edges of color $\col{x}$, which matches their budget.
	
	\proofsubparagraph{($\boldsymbol{\Leftarrow}$)}
	Assume that \linL is a solution for the instance \instance of \capacitatedQueue.
	We now construct a variable assignment $\Psi\colon X \to \{0,1\}$ and consider a variable $x \in X$.
	We set $x$ to true if and only if $\sigma(s_xt_x) = 2$.
    Consider an arbitrary clause $c = (x \lor y \lor z) \in C$.
    By \Cref{lem:queue-layouts-components-size-clause}, we can find two literals $p,q\in \{x,y,z\}$ such that $\sigma(s_pt_p) = \sigma(a_{p}^cb_{p}^c) \neq \sigma(a_{q}^cb_{q}^c) = \sigma(s_qt_q)$. 
    Thus, there are two variables $p,q \in \{x, y, z\}$ such that $\Psi(p) \neq \Psi(q)$, i.e., they witness that the clause $c$ is satisfied under $\Psi$.
    Since $c$ was selected arbitrarily, we conclude that $\Psi$ is a solution to $\varphi$.
\end{prooflater}

\subparagraph*{\textsf{FPT}-Algorithm for Two Pages.}
Let $\instance = \instanceLongQueueTwoPages$ be an instance of \capacitatedQueue and note that $\ell = 2$, which implies $m \in \BigO{n}$ for positive instances~\cite{DBLP:journals/dmtcs/DujmovicW04}.
As for \Cref{thm:stack-layouts-components-fpt}, from $\ell = 2$ we can conclude that every connected component of the conflict graph $\queueConflict[G, \prec]$ must be bipartite; edges are assigned to the same page if their corresponding vertices are in the same bipartition.
If the number of connected components is bounded, we can branch over all valid page assignments.
The proof of the following theorem is identical to \Cref{thm:stack-layouts-components-fpt} (by using $\queueConflict$ instead of $\stackConflict$) and thus omitted.

\begin{theorem}
	\label{thm:queue-layouts-two-pages-components-fpt}
    \capacitatedQueueLong with $\ell = 2$ pages can be solved constructively in $\BigO{2^{\conComp}\cdot (n + \chi) + n^2}$ time (which is \FPT\ in $\conComp = \Size{\mathcal{C}(\queueConflict[G, \prec])}$).
\end{theorem}

\subparagraph*{Hardness for Three Pages.}
We now show that we cannot generalize \Cref{thm:queue-layouts-two-pages-components-fpt} to more than two pages.
More concretely, we show that \capacitatedQueue is \NP-complete for instances $\instanceLongQueue$ with $\ell = 3$ and where $\queueConflict$ is connected.
This rules out an \XP-algorithm with respect to the number of pages $\ell$ and the number of connected components $\conComp$.

\begin{restatable}\restateref{thm:queue-layouts-three-pages-connected}{theorem}{theoremQueueThreePagesConnected}
	\label{thm:queue-layouts-three-pages-connected}
	\capacitatedQueueLong is \NP-complete even 
    when restricted to instances $\instanceLongQueue$ with $\ell = 3$ and where $\queueConflict[G, \prec]$ is connected.
\end{restatable}
\begin{proofsketch}
    Membership in \NP\ is immediate.
    To show \NP-hardness, we take an instance $\instance = \instanceLongQueueTwoPages$ of \capacitatedQueue with $\ell = 2$ and where $\sum_{c \in [\chi]} (b_1(c) + b_2(c)) = \Size{E}$.
    Recall that \capacitatedQueue with $\ell = 2$ is \NP-complete by \Cref{thm:queue-layouts-components-size}.
    Let $\chi$ be the number of colors in \instance.
    We create an instance $\instance'$ with $\ell = 3$ by adding an edge $uv$ of color $\chi + 1$ to~$G$ that spans the entire instance \instance.
    Since $uv$ spans every other edge, the resulting conflict graph~$\queueConflict'$ is connected.
    We can set the budgets such that $uv$ is the only edge on the third page and no other edge can be assigned to that page.
    Thus, the instances $\instance$ and $\instance'$ are equivalent on the first two pages, i.e., \instance is a positive instance if and only if $\instance'$ is.
\end{proofsketch}
\begin{prooflater}{ptheoremQueueThreePagesConnected}
We focus on \NP-hardness since membership in \NP\ is immediate.
Let $\instance = \instanceLongQueueTwoPages$ be an instance of \capacitatedQueue with $\ell = 2$ and where $\sum_{c \in [\chi]} (b_1(c) + b_2(c)) = \Size{E}$.
Recall from \Cref{thm:queue-layouts-components-size} that this problem setting is already \NP-complete.
Let $\chi$ be the number of colors in \instance.
We introduce an edge $uv$ and set $\delta(uv) = \chi + 1$.
Moreover, we prepend $u$ to $\prec$ and append $v$ to $\prec$.
More formally, the new spine order $\prec'$ is obtained by taking the transitive closure of $u \prec' w \prec' v$ for every $w \in V(G)$ and $a \prec' b$ for every $a, b \in V(G)$ with $a \prec b$.
To fix the budgets $b_i'$, $i = 1,2,3$, we set $b_i'(c) = b_i(c)$ for $i = 1,2$ and $b_3'(c) = 0$ for every $c \in [\chi]$.
Furthermore, we set $b_1'(\chi + 1) = b_2'(\chi + 1) = 0$ and $b_3'(\chi + 1) = 1$.
Let $\instance' = (G', \delta', \prec', \ell' = 3, (b_i')_{i \in [\ell']})$ be the resulting instance of \capacitatedQueue.
We can observe that it can be constructed in polynomial time. 
The edge $uv$ nests all other edges and the corresponding vertex is, therefore, adjacent to every other vertex in $\queueConflict[G', \prec']$, i.e., $\queueConflict[G', \prec']$ is connected. 
The instances $\instance$ and $\instance'$ are equivalent up to the edge $uv$, for which the budgets ensure that it must be assigned to the third page and no other edge can be assigned to that page.
Consequently, \instance is a positive instance if and only if $\instance'$ is, which completes the proof.
\end{prooflater}

\subsection{Hardness for Monochromatic Instances}
\label{sec:queue-layouts-monochromatic}
With \Cref{thm:queue-layouts-components-size}, we have shown that \capacitatedQueue is \NP-complete even for fairly restricted instances.
However, the hardness reduction employed in \Cref{thm:queue-layouts-components-size} constructs instances with a linear number of colors.
Therefore, it is natural to ask for the complexity when parameterized by the number of colors $\chi$.
In this section, we show that \capacitatedQueue remains \NP-complete even when $\chi=1$, i.e., every edge of $G$ has the same color, which we call \emph{monochromatic} instances, and, as we see later, \W[1]-hard with respect to $\ell$.\footnote{We note that a similar reduction was given by Gomes et al.~\cite{DBLP:journals/dmtcs/GomesLS19} to show that \probname{Equitable Coloring} is
\W[1]-hard parameterized by the number of colors.
We expect that their reduction can be adapted to our setting, but present for completeness a self-contained reduction that furthermore allows additional edge-colors or different budgets, which we exploit in \Cref{thm:queue-layouts-two-colors-connected-conflict,cor:queue-layouts-monochromatic-pages-connected}.}

\subparagraph*{Overview of the Reduction.}
We reduce from the problem \probname{Bin Packing} defined as follows.
Given a set $S$ of $\eta$ items, each item $x \in S$ associated with a weight $w(x) \in \mathbb{N}$, and $k$ bins of capacity $B \in \mathbb{N}$, can $S$ be partitioned into $k$ subsets $T_1, \ldots, T_k$ such that $\sum_{x \in T_i} w(x) \leq B$ for every $i \in [k]$?
\probname{Bin Packing} is known to be strongly \NP-complete and thus we can assume that the weight of each item $x \in S$ and the capacity $B$ are polynomially bounded in $\eta$~\cite{GJ.CIG.1979}.
Moreover, \probname{Bin Packing} remains \NP-complete even if $\sum_{x \in S} w(x) = k\cdot B$, i.e., if every bin must be filled up to its capacity $B$ since we can simply introduce sufficiently many dummy items of weight one.
On a high-level, in our reduction we represent every bin by a page.
Assigning edges to that page should correspond to placing the items in a bin.
To ensure correctness, we encode every item $x \in S$ with a gadget consisting of two components.
The first component consists of $k - 1$ pairwise nesting edges, a \emph{($k-1$)-rainbow}, which must, therefore, be distributed among $k - 1$ pages.
The second component consists of a star that represents the weight of $x$.
By placing it below the ($k-1$)-rainbow, all edges of the star must be placed on the unique page that does not contain an edge of the rainbow.
Since the instance is monochromatic, we cannot use colors to distinguish edges from the star, which should count towards our bin-capacity, from those of the rainbow, whose sole purpose is to ensure correctness of the reduction.
To circumvent this limitation, we scale the weight of every item by a large number to make it distinguishable from the rainbow edges.

\subparagraph*{Item Gadget.}
Let $x \in S$ be an item.
The item gadget for $x$ contains $k - 1$ \emph{rainbow edges} $s_x^it_x^i$, $i \in [k - 1]$, and we order their endpoints along $\prec$ as follows.
We set $s_x^i \prec t_x^j$ for every $i,j \in [k - 1]$.
Moreover, we set $s_x^i \prec s_x^{i + 1}$ and $t_x^{i + 1} \prec t_x^{i}$ for every $i \in [k - 2]$.
\begin{figure}
	\centering
	\includegraphics[page=1]{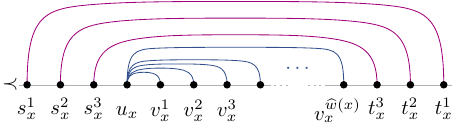}
	\caption{The item gadget for an item $x \in S$ and with $k = 4$. Although the constructed instance is monochromatic, we color edges from the rainbow and the item purple and blue, respectively.}
	\label{fig:queue-layout-item-gadget}
\end{figure}

We define $\widehat{w}(x) \coloneqq k \cdot \eta \cdot w(x)$ to be the \emph{scaled weight} of the item $x$.
Observe that the scaled weight of $x$ is still polynomially bounded in $\eta$.
The \emph{value star} of the item gadget is a star $K_{1,\widehat{w}(x)}$ with center $u_x$ and leaves $v_x^i$, $i \in [\widehat{w}(x)]$.
We place the value star below the rainbow edges, i.e., set $s_x^{k - 1} \prec u_x \prec v_x^1 \prec v_x^2 \prec \ldots \prec v_x^{\widehat{w}(x)} \prec t_x^{k - 1}$ and take the transitive closure of this and the above-defined relations; see the blue edges in \Cref{fig:queue-layout-item-gadget}.
Since the rainbow edges are pairwise nesting and each rainbow edge nests the value star, we observe:
\begin{observation}
	\label{obs:queue-layouts-monochromatic-item-gadget}
	In every $k$-page queue layout \linL of the item gadget for an item $x \in S$ we have $\sigma(s_x^it_x^i) \neq \sigma(s_x^jt_x^j)$ for every $i,j \in [k - 1]$, $i \neq j$.
	Moreover, we have $\sigma(s_x^it_x^i) \neq \sigma(u_xv_x^j)$ for every $i \in [k - 1]$ and $j \in [\widehat{w}(x)]$ and thus $\sigma(u_xv_x^i) = \sigma(u_xv_x^j)$ for every $i,j \in [\widehat{w}(x)]$.
\end{observation}

Assigning the value star to page $p \in [k]$ should correspond to placing the item in the $p$th bin.
\Cref{obs:queue-layouts-monochromatic-item-gadget} ensures that this equivalence between the two problems is well-defined.

\subparagraph*{Putting Everything Together.}
With the item gadget at hand, we are now ready to describe the complete reduction.
Let $(S, w, k, B)$ be an instance of \probname{Bin Packing} with items $S = \{x_1, \ldots, x_{\eta}\}$, $\eta = \Size{S}$.
For every item $x_i$, we construct the item gadget and order all item gadgets on the spine based on the index of the item they represent.
\shortLong{
Moreover, we introduce \emph{dummy} edges which we place at the end of the spine such that they are not in conflict with each other or any edge from an item gadget.
These dummy edges will help us to satisfy the budget on every page.
To this end, observe that every page will receive $k \cdot \eta \cdot B$ non-rainbow edges from the item gadgets in the intended solution.
Moreover, in the worst case a bin will host a single item, which implies that it receives $\eta - 1$ rainbow edges for the items assigned to the other bins.
Hence, we have to set a budget of $b_p(1) = k \cdot \eta \cdot B + (\eta - 1)$ for every page $p \in [k]$ to ensure correctness; recall that the instance we construct is monochromatic.
To simplify our arguments (and strengthen our result), we want that every page receives $b_p(1)$ edges.
There are in total $\eta \cdot (k - 1)$ rainbow edges, and we need in total $k \cdot (\eta - 1)$ additional edges to fill up the budgets. 
This gives an imbalance of precisely $k \cdot (\eta - 1) - (\eta \cdot (k - 1)) = \eta - k$ edges (we can assume $\eta \geq k$), which will be the number of dummy edges that we introduce.%
}{
Moreover, we introduce $\eta - k$ \emph{dummy} edges which will help us to satisfy the budget on every page.
We place them at the end of the spine such that they are not in conflict with each other or any edge from an item gadget.
Finally, we give every edge $e$ the color $\delta(e) = 1$ and set for each page $p \in [k]$ the budget as $b_p(1) = k \cdot \eta \cdot B + (\eta - 1)$.
}
\Cref{fig:queue-layout-bin-packing-example} shows a small example of our reduction.
We now establish its correctness.

\begin{figure}
	\centering
	\includegraphics[page=1]{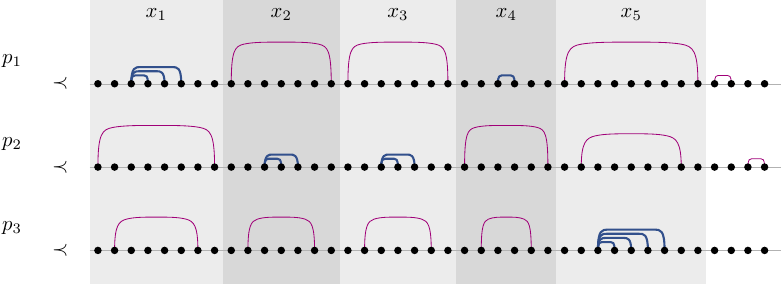}
	\caption{The constructed instance of \capacitatedQueue for an instance $(\{x_1, x_2, x_3, x_4, x_5\}, w, 3, 4)$ of \probname{Bin Packing} with $w(x_1) = 3$, $w(x_2) = w(x_3) = 2$, $w(x_4) = 1$, and $w(x_5) = 4$. For illustration purposes, we color rainbow and item edges purple and blue, respectively.
		Each bold edge represents $k\cdot\eta = 15$ edges. 
    Gray bands indicate the different items and pages are stacked vertically.}
	\label{fig:queue-layout-bin-packing-example}
\end{figure}

\begin{restatable}\restateref{thm:queue-layouts-monochromatic}{theorem}{theoremQueueMonochromatic}
	\label{thm:queue-layouts-monochromatic}
	\capacitatedQueueLong is \NP-complete even for instances that are balanced, tight, and monochromatic. %
\end{restatable}
\begin{proofsketch}
	Let $(S = \{x_1, \ldots, x_{\eta}\}, w, k, B)$ be an instance of \probname{Bin Packing} such that $\sum_{x \in S} {w(x)} = k\cdot B$ and $B$ and $w(x)$ are polynomially bounded in $\eta$ for every $x \in S$.
    Moreover, we assume without loss of generality $\eta \geq k \geq 2$.
    Let $\instance$ be the instance of \capacitatedQueue that we construct.
    We focus in this sketch on the backwards direction of the correctness argument.
    
	Assume that $\instance$ admits a solution $\linL$.
    We assign an item $x \in S$ to bin $T_i$ if every edge of the star $K_{1, \widehat{w}(x)}$ is on page $i$, for all $i \in [k]$.
    \Cref{obs:queue-layouts-monochromatic-item-gadget} ensures that this is well-defined.
    Assume that $T_1, \ldots, T_k$ is not a solution of $(S, w, k, B)$, i.e., there exists a bin $T_i$ such that $\sum_{x \in T_i} w(x) \neq B$.
    First, assume $\sum_{x \in T_i} w(x) \geq B+ 1 > B$.
    By the construction of the bin $T_i$, there are at least $\sum_{x \in T_i} \widehat{w}(x)$ edges assigned to page $i$.
    It holds $\sum_{x \in T_i} \widehat{w}(x) = \sum_{x \in T_i} (k \cdot \eta \cdot w(x)) \geq k \cdot \eta \cdot (B + 1) > k \cdot \eta \cdot B + (\eta - 1) = b_i(1)$.
    This contradicts the assumption that \linL is a solution.
    Second, assume $\sum_{x \in T_i} w(x) < B$.
    By the pigeonhole principle, there is a different bin $T_j$, $j \in [k]$, $j \neq i$, such that $\sum_{x \in T_j} w(x) > B$.
    The above arguments lead to a contradiction.
    We conclude that $T_1, \ldots, T_k$ is a solution.
\end{proofsketch}
\begin{prooflater}{ptheoremQueueMonochromatic}
	Let $(S = \{x_1, \ldots, x_{\eta}\}, w, k, B)$ be an instance of \probname{Bin Packing} such that $\sum_{x \in S} {w(x)} = k\cdot B$ and $B$ and $w(x)$ are polynomially bounded in $\eta$ for every $x \in S$.
    In the following, we assume without loss of generality that $\eta \geq k \geq 2$ (otherwise, we can return a trivial positive or negative instance of \capacitatedQueue).
	Moreover, let $\instance = \instanceLongQueue[k]$ be the instance of \capacitatedQueue that we construct.
	Note that \instance can be constructed in polynomial time since $w(x)$ is polynomially bounded in $\eta$ for every $x \in S$.
	Moreover, we can observe that \instance is monochromatic, and we have $\ell = k$ and $b_p(1) = b_q(1)$ for every two pages $p,q \in [\ell]$.
	It remains to show correctness of the reduction.
	
	\proofsubparagraph{($\boldsymbol{\Rightarrow}$)}
	Assume that $(S, w, k, B)$ is a positive instance of \probname{Bin Packing} and let $T_1, \ldots, T_k$ be a partitioning of $S$ that witnesses this.
	We construct a linear layout \linL of \instance as follows.
	Consider a bin $T_i$, $i \in [k]$.
	Let $x \in T_i$ be an item in the $i$th bin and consider the corresponding item gadget.
	First, we assign all edges of the star $K_{1,\widehat{w}(x)}$ to page $i$, i.e., we set $\sigma(u_xv_x^j) = i$, $j \in [\widehat{w}(x)]$.
	Next, distribute the $k - 1$ rainbow edges in the gadget among the remaining $k - 1$ pages, i.e., for every page $p \in [k] \setminus \{i\}$ there is precisely one $q \in [k - 1]$ such that $\sigma(s_x^qt_x^q) = p$.
	Observe that this resembles the necessary page assignment from \Cref{obs:queue-layouts-monochromatic-item-gadget}.
	In particular, no two rainbow edges from the gadget are assigned to the same page or to the page of the star $K_{1,\widehat{w}(x)}$ (which itself does not introduce a nesting).
	Since $\sum_{x \in S} {w(x)} = k\cdot B$ and $T_1, \ldots, T_k$ is a solution, we have $\sum_{x \in T_i} w(x) = B$.
	Overall, we have assigned $\sum_{x \in T_i} {\widehat{w}(x)} = k \cdot \eta \cdot B $ edges to page $i$ so far.
	Doing the same procedure for the other bins $T_j$, $j \neq i$, we assign $\eta - \Size{T_i}$ edges in addition to page $i$.
	These account for the items that are placed in some bin $T_j$, $j \neq i$, for which we will assign one rainbow edge to page $i$.
	We ``fill up'' page $i$ to the budget $b_i(1)$, by assigning $\Size{T_i} - 1$ dummy edges to page $i$; we argue in the end that there are sufficiently many dummy edges.
	Furthermore, recall that they do not nest any other edge.
	We repeat these steps for every item in the bin $T_i$, and for every bin of the solution to eventually obtain a linear layout \linL.
	Thanks to the above arguments, \linL is an admissible queue layout.
	It remains to argue that there are enough dummy edges to fill up every page.	
	We assign $k \cdot \eta \cdot B$ non-rainbow edges from the item gadgets to every page.
	Thus, every page ``needs'' $\eta - 1$ additional edges.
	There are in total $\eta \cdot (k - 1)$ rainbow edges, and we need in total $k \cdot (\eta - 1)$ to fill up the page budgets.
	This gives an imbalance of precisely $k \cdot (\eta - 1) - (\eta \cdot (k - 1)) = \eta - k$ edges, which is the number of dummy edges.
	Combining all, we conclude that \instance is a positive instance of \capacitatedQueue.

	\proofsubparagraph{($\boldsymbol{\Leftarrow}$)}
	Assume that $\instance$ admits a solution $\linL$.
    We now construct a partitioning of $S$ into $k$ bins $T_1, \ldots, T_k$.
	Consider an item $x \in S$ and the corresponding item gadget.
    Due to \Cref{obs:queue-layouts-monochromatic-item-gadget}, there exists a page $i \in [k]$ that contains every edge of the star $K_{1, \widehat{w}(x)}$.
    We place the item in the bin $T_i$ and repeat this process for the remaining items of $S$.
    Towards arguing that $T_1, \ldots, T_k$ is a solution of $(S, w, k, B)$ we assume that it is not.
    Since we have $\sum_{x \in S} w(x) = k \cdot B$, this means there exists a bin $T_i$ such that $\sum_{x \in T_i} w(x) \neq B$.
    First, assume that $\sum_{x \in T_i} w(x) \geq B+ 1 > B$.
    By construction, this implies that we have assigned at least $\sum_{x \in T_i} \widehat{w}(x)$ edges to page $i$.
    Observe that we have $\sum_{x \in T_i} \widehat{w}(x) = \sum_{x \in T_i} (k \cdot \eta \cdot w(x)) \geq k \cdot \eta \cdot (B + 1)$.
    However, the budget of page $i$ is $b_i(1) = k \cdot \eta \cdot B + (\eta - 1) < k \cdot \eta \cdot B + k \cdot \eta$ (recall that we have $k \geq 2$).
    This contradicts the assumption that \linL is a solution.
    Hence, a bin $T_i$ with $\sum_{x \in T_i} w(x) > B$ cannot exist.
    Now, assume that there is a bin $T_i$ with $\sum_{x \in T_i} w(x) < B$.
    From $\sum_{x \in S} w(x) = k \cdot B$ and the pigeonhole principle, this means that there is a different bin $T_j$, $j \in [k]$, $j \neq i$, such that $\sum_{x \in T_j} w(x) > B$.
    However, by the above reasoning this cannot be the case.
    Combining all, we conclude that $T_1, \ldots, T_k$ is a solution of $(S, w, k, B)$.
\end{prooflater}

\subparagraph*{Consequences of \Cref{thm:queue-layouts-monochromatic}.}
We can observe that the hardness-construction behind \Cref{thm:queue-layouts-monochromatic} creates instances that have as many pages as the instance $(S, w, k, B)$ of \probname{Bin Packing} has bins, i.e., $k$.
This, together with the fact that the hardness already applies for monochromatic instances, allows us to derive hardness results with respect to several parameter combinations, which we discuss in the following.
First, note that unary \probname{Bin Packing} is \W[1]-hard with respect to the number of bins $k$~\cite{JKMS.Bpf.2013}:
\begin{corollary}
    \label{cor:queue-layouts-monochromatic-pages}
    \capacitatedQueueLong is \W\textup{[1]}-hard when parameterized by the number of pages $\ell$ even when restricted to balanced and monochromatic instances.
\end{corollary}

Second, we can use the same ``trick'' as for \Cref{thm:queue-layouts-three-pages-connected} to turn the conflict graph $\queueConflict$ into a connected graph.
To this end, we add a new edge $uv$ of color $\delta(uv) = 2$ that spans the entire instance, i.e., which ensures that $\queueConflict$ is connected.
By adding a page $k + 1$ and adjusting the budget of every page $p \in [k + 1]$ such that precisely the additional edge $uv$ is placed on page $k + 1$, we preserve equivalence between the two instances:
\begin{restatable}\restateref{thm:queue-layouts-two-colors-connected-conflict}{theorem}{theoremQueueColorsConnected}
    \label{thm:queue-layouts-two-colors-connected-conflict}
    \capacitatedQueueLong is \NP-complete even when restricted to tight instances \instanceLongQueue where $\chi = 2$ and $\queueConflict(G, \prec)$ is connected.
\end{restatable}
\begin{prooflater}{ptheoremQueueColorsConnected}
    As already observed several times, \NP-membership is immediate.
    Towards showing \NP-hardness, recall first that \capacitatedQueue\ is \NP-complete even when restricted to monochromatic instances by \Cref{thm:queue-layouts-monochromatic}.
    Let $\instance = \instanceLongQueue$ be such a monochromatic instance.
    We now construct an instance $\instance' = (G', \delta', \prec', \ell' = \ell + 1, (b_i')_{i \in [\ell']})$ of \capacitatedQueue as follows.
    For $G'$, we take $G$ and introduce two new vertices $u$ and $v$ which are adjacent.
    We set $\delta'(uv) = 2$ and $\delta'(e) = \delta(e)$ for all remaining edges of $G'$.
    The new spine order $\prec'$ is obtained by taking the transitive closure of $u \prec' w \prec' v$ for every $w \in V(G)$ and $a \prec' b$ for every $a, b \in V(G)$ with $a \prec b$.
    For the budgets, we set $b_p'(1) = b_p(1)$ and $b_p'(2) = 0$ for every page $p \in [\ell]$.
    Moreover, we set $b_{\ell + 1}'(1) = 0$ and $b_{\ell + 1}'(2) = 1$.

    Since $uv$ nests every edge $e \in E(G)$, $\queueConflict[G', \prec']$ is connected. 
    $\instance'$ can be constructed in polynomial time.
    Moreover, by the definition of the budgets $b'$, we can transform every solution $\linL$ of \instance to a solution $\langle \prec', \sigma'\rangle$ of $\instance'$ by assigning the edge $uv$ to page $\ell + 1$.
    Symmetrically, a solution $\langle \prec', \sigma'\rangle$ of $\instance'$ can be transformed into a solution $\linL$ of \instance by dropping the edge $uv$, which shows correctness and completes the proof.
\end{prooflater}

Third, by applying (the reasoning that led to) \Cref{cor:queue-layouts-monochromatic-pages} to \Cref{thm:queue-layouts-two-colors-connected-conflict} and observing that we only need one additional page to ensure that \queueConflict is connected, we can conclude:
\begin{corollary}
    \label{cor:queue-layouts-monochromatic-pages-connected}
    \capacitatedQueueLong is \W\textup{[1]}-hard when parameterized by the number of pages $\ell$ even when restricted to instances \instanceLongQueue where $\chi = 2$ and $\queueConflict(G, \prec)$ is connected.
\end{corollary}

\subsection{Parameterizing by the Number of Pages and Colors}
\label{sec:queue-layouts-pages-colors}
In this section, we show that \capacitatedQueue is in \XP\ when parameterized by the number of pages $\ell$ and the number of colors $\chi$.
By \Cref{cor:queue-layouts-monochromatic-pages}, this result is the ``best possible'' we can hope for.
Moreover, \capacitatedQueue is para\NP-hard when parameterized by $\ell$ and $\chi$ alone; recall \Cref{thm:queue-layouts-components-size,thm:queue-layouts-monochromatic}.
On a high-level, we employ a dynamic program (DP), which considers the edges in lexicographic order $<$ according to $\prec$, i.e., for two edges $uv, xy \in E(G)$ with $u \prec v$ and $x \prec y$ we have $uv < xy$ if and only if $u \prec x$, or $u = x$ and $v \prec y$.
When processing the next edge $uv$, we consider every possible page assignment for $uv$.
Processing the edges in lexicographic order has the algorithmic advantage that we can easily prevent two nesting edges as the following lemma underlines; see also \Cref{fig:queue-layout-nesting}.
\begin{figure}
	\centering
	\includegraphics[page=1]{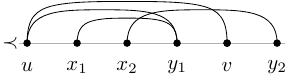}
	\caption{Visualization of \Cref{lem:queue-layouts-pages-colors-nesting}. Observe that only the edges $uv$ and $x_1y_1$ are nesting. The edges $uv$ and $uy_1$ are not nesting although $y_1 \prec v$. Note that $y_1 \prec v$ implies $uy_1 < uv$.}
	\label{fig:queue-layout-nesting}
\end{figure}

\begin{restatable}\restateref{lem:queue-layouts-pages-colors-nesting}{lemma}{lemmaQueuePagesColorsNesting}
	\label{lem:queue-layouts-pages-colors-nesting}
	Let $uv, xy \in E$ be two edges such that $u \prec v$, $x \prec y$ and $uv < xy$.
	The edges $uv$ and $xy$ are nesting if and only if $y \prec v$.
\end{restatable}
\begin{proofsketch}
	For the forward direction, observe that a nesting implies $u \prec x \prec y \prec v$ or $x \prec u \prec v \prec y$.
    However, only the former is compatible with $uv < xy$.
    For the backward direction, we use $uv < xy$ to get $u \preceq x$.
    Since $u = x$ is impossible because of $y \prec v$, we have $u \prec x$, which together with $u \prec v$ and $x \prec y$ gives $u \prec x \prec y \prec v$, i.e., the edges nest.
\end{proofsketch}
\begin{prooflater}{plemmaQueuePagesColorsNesting}
	Let $uv, xy \in E$ be two edges as in the statement.
	
	\proofsubparagraph*{($\boldsymbol{\Rightarrow}$)}
	Assume that the two edges are nesting, which means that we have $u \prec x \prec y \prec v$ or $x \prec u \prec v \prec y$.
	Only the former is possible by $uv < xy$, which completes this direction.
	
	\proofsubparagraph*{($\boldsymbol{\Leftarrow}$)}
	Assume that we have $y \prec v$.
	From $uv < xy$ we get that $u \preceq x$ must hold.
	Observe that $u = x$ is not possible since $u = x$ and $y \prec v$ would imply $xy < uv$.
	Thus, we have $u \prec x$ and together with $x \prec y$ and $y \prec v$, this implies $u \prec x \prec y \prec v$, i.e., the two edges nest.
\end{prooflater}

We can now present our \XP-algorithm, which can be seen as a more involved variant of the classical polynomial-time algorithm for fixed-order queue layouts~\cite{HR.LOG.1992}.

\begin{restatable}\restateref{thm:queue-layouts-pages-colors-xp}{theorem}{theoremQueuePagesColors}
	\label{thm:queue-layouts-pages-colors-xp}
    \capacitatedQueueLong can be solved constructively in $\BigO{\ell\cdot {(m + 1)}^{(\ell \cdot \chi) + 1} \cdot (n + 1)^{\ell + 1}}$ time (which is \XP\ in $\ell + \chi$).
\end{restatable}
\begin{proofsketch}
    We assume that the edges $E=\{e_1, \ldots, e_m\}$ are indexed according to their lexicographic order $<$ (which we can compute in $\BigO{m \log m}$ time).
	Furthermore, we index the vertices $V = \{v_1, \ldots, v_n\}$ according to $\prec$.
    We employ a dynamic programming (DP) approach that considers the edges $e_1, \ldots, e_m$ according to the lexicographic order $<$ and evaluates all potential page assignments for each new edge.
    To this end, we consider \emph{states} $I$.
    A state $I$ stores information about the currently processed edge $e(I)$.
    Moreover, it stores the number $x(I, p, c)$ of edges of color $c \in [\chi]$ which are currently assigned to page $p$.
    Finally, it keeps track of the rightmost vertex $v(I, p)$ which is incident to an edge on page $p$.
    We write $v(I, p) = \bot$ if there is no edge assigned to page $p$.
    Observe that the number of states is in $\BigO{m \cdot {(m + 1)}^{\ell \cdot \chi} \cdot (n + 1)^{\ell}}$: the graph $G$ contains $m$ edges, $x(I, p, c) \in [m]_0$ for every page $p \in [\ell]$ and color $c \in [\chi]$, and there are $n + 1$ possible values for $v(I, p)$ for every page $p \in [\ell]$.

    Let $G_i$ be the graph $G_i \coloneqq (V, \{e_1, \ldots, e_i\})$. 
    We use a binary DP-table $D$ to store for each state $I$ whether it is \emph{feasible}, i.e., whether there exists an $\ell$-page queue layout $\langle \prec, \sigma_i\rangle$ of~$G_i$, where $i$ is the index of the edge $e(I)$, such that there are precisely $x(I, p, c)$ edges of color~$c$ assigned to page $p$, for every $c \in [\chi]$ and $p \in [\ell]$, and $v(I, p)$ is the rightmost vertex which is incident to an edge on page $p$, for every $p \in [\ell]$.
    Such a layout \emph{witnesses} the feasibility of $I$.

	To compute $D[I]$, we assume that $D[I']$ has already been computed for every state $I'$ where $e(I') = e_{j}$, $e(I) = e_i$, and $j < i$.
	Consider the base case where $e(I) = e_1 = uv$, $u \prec v$.
	We set $D[I] = 1$ if and only if $I$ models the assignment of $e_1$ to some page $p \in [\ell]$.
    In particular, we require $x(I, q, c) = 1$ if $c = \delta(e_1)$ and $q = p$ and $x(I, q, c) = 0$ otherwise.
    Moreover, we require $v(I, q) = v$ if $q = p$ and $v(I, q) = \bot$ otherwise.
    In the general case with $e(I) = e_i = uv$, $u \prec v$, for some $i > 1$, we set $D[I] = 1$ if and only if there exists a state~$I'$ with $e(I') = e_{i - 1}$ and $D[I'] = 1$ such that $I$ can be obtained from $I'$ by assigning $uv$ to some page $p \in [\ell]$.
    In particular, this means $x(I, q, c) = x(I', q, c) + t$ where $t = 1$ if $c = \delta(e_i)$ and $q = p$ and $t = 0$ otherwise.
    To ensure that $uv$ does not nest with some other edge on page $p$, we make use of \Cref{lem:queue-layouts-pages-colors-nesting} and require $v(I', p) = \bot$ or $v(I', p) = w$ for some $w \in V$ with $w \preceq v$.
    Of course, also $v(I, q)$ must be consistent with the assignment of $uv$ to page $p$, i.e., $v(I, q) = v$ if $q = p$ and $v(I, q) = v(I', q)$ otherwise.
    Using induction on (the order of) $e(I)$, we can show that we have $D[I] = 1$ if and only if $I$ is feasible (see also \Cref{claim:queue-layouts-pages-colors-xp-correctness} in \Cref{app:queue-layouts-pages-colors}).
    While for the base case, i.e., for states with $e(I) = e_1$, this is immediate, we now sketch the inductive step with $e(I) = e_i$, $i > 1$.
    Assume that the statement holds for all states $I$ with $e(I) = e_{j}$ for $j < i$.
    
    ($\Rightarrow$) If $I$ is feasible, there is a witnessing $\ell$-page queue layout $\langle \prec, \sigma_i\rangle$ of $G_i$.
    Consider the layout $\langle \prec, \sigma_{i - 1}\rangle$ that we get by removing the edge $e_i$ from $\langle \prec, \sigma_i\rangle$.
    Since queue layouts are preserved under edge deletion, it is an $\ell$-page queue layout of $G_i - \{e_i\} = G_{i - 1}$ which  witnesses the feasibility of the corresponding state $I'$.
    Note that we consider the state $I'$ when evaluating $D[I]$ as $e(I') = e_{i - 1}$ and by induction hypothesis we have $D[I'] = 1$.
    A careful examination of our criteria for $D[I] = 1$ shows that this implies $D[I] = 1$.
    
    ($\Leftarrow$) If $D[I] = 1$, we consider the state $I'$ and the page $p \in [\ell]$ which we used to conclude $D[I] = 1$.
    Let $\langle \prec, \sigma_{i - 1}\rangle$ be a witnessing $\ell$-page queue layout of $G_{i - 1}$ and observe that it must exist by the induction hypothesis (as $e(I') = e_{i - 1}$).
    We extend $\langle \prec, \sigma_{i - 1}\rangle$ by assigning the edge $e_i$ to page $p$.
    Since we require $v(I', p) = \bot$ or $v(I', p) = w$ for some $w \in V$ with $w \preceq v$ and consider the edges in lexicographic order, the resulting layout $\langle \prec, \sigma_{i}\rangle$ is an $\ell$-page queue layout of $G_{i - 1} + \{e_i\} = G_i$.
	By also examining the number of edges of each color on each page and the respective rightmost vertices, we get that $\langle \prec, \sigma_{i}\rangle$ witnesses the feasibility of $I$.

    Thus, we have $D[I] = 1$ if and only if $I$ is feasible and \instance has a solution if and only if there is a state $I$ with $e(I) = e_m$ and $D[I] = 1$ such that $x(I, p, c) \leq b_p(c)$ for every page $p \in [\ell]$ and color $c\in [\chi]$.
    We argue the running time of our DP in the full proof (see \Cref{app:queue-layouts-pages-colors}) and conclude this sketch by 
    noting that we can obtain a solution (if there exists one) via standard backtracking techniques.%
\end{proofsketch}
\begin{prooflater}{ptheoremQueuePagesColors}
	For the following arguments, fix an instance $\instance = \instanceLongQueue$ of \capacitatedQueue.
	We first order the edges of $G$ lexicographically according to $\prec$.
	This preprocessing step can be done in $\BigO{m \log m}$ time.
	Let $<$ denote the resulting order of the edges.
	For ease of presentation, we re-index the edges $E=\{e_1, \ldots, e_m\}$ of $G$ according to $<$, i.e., for every two edges $e_i,e_j \in E$ we have $e_i < e_j$ if and only if $i < j$.
	Furthermore, we index the vertices $V = \{v_1, \ldots, v_n\}$ of $G$ according to $\prec$, i.e., for every two vertices $v_i,v_j \in V$ we have $i < j$ if and only if $v_i \prec v_j$.
	Call this first step the \emph{preprocessing} step.
	
	\proofsubparagraph{The Algorithm.}
	We employ a dynamic programming (DP) approach that considers the edges according to the lexicographic order $<$ and evaluates all potential page assignments for each new edge.
	Thanks to \Cref{lem:queue-layouts-pages-colors-nesting} it is sufficient to store for each page, in addition to the number of edges on that page of each color, the rightmost vertex which is incident to an edge on that page.
	More formally, we consider \emph{states}, which are defined by an edge $e_i \in E$, a number $x_p^c \in [m]_0$ for every page $p \in [\ell]$ and color $c \in [\chi]$, and up to one vertex $v_p \in V \cup \{\bot\}$ for every page $p \in [\ell]$.
	For a state $I$, we use $e(I)$ to address the edge $e_i$, $x(I, p, c)$ to address $x_p^c$ for $p \in [\ell]$ and $c \in [\chi]$, and $v(I, p)$ to address the vertex ${v_p}$ for $p \in [\ell]$ or $\bot$ if there is no such vertex ${v_p}$.
	Let $G_i$ be the graph $G_i \coloneqq (V, E_i)$ where $E_i \coloneqq \{e_1, \ldots, e_i\}$.
	We call a state $I$ with $e(I) = e_i$ \emph{feasible} if and only if there exists an $\ell$-page queue layout $\langle \prec, \sigma_i\rangle$ of $G_i$ such that there are precisely $x(I, p, c)$ edges of color $c$ assigned to page $p$, for every $c \in [\chi]$ and $p \in [\ell]$, and $v(I, p)$ is the rightmost vertex which is incident to an edge on page $p$, for every $p \in [\ell]$.
	If no edge is assigned to a page $p$, then we require $v(I,p) = \bot$.
	In our algorithm, we maintain a binary DP-table $D$ which should store, for every state $I$, whether it is feasible.
	
	Consider a state $I$.
	We now describe how to compute $D[I]$, assuming that $D[I']$ has already been computed for every state $I'$ where $e(I') = e_{j}$, $e(I) = e_i$, and $j < i$.
	First, we consider the base case where $e(I) = e_1 = uv$, $u \prec v$.
	We set $D[I] = 1$ if and only if there exists a page $p \in [\ell]$ such that
	\begin{enumerate}[B1.]
		\item $x(I, p, \delta(e_1)) = 1$,\label{dp:queue-layouts-pages-colors-xp-base-start}
		\item $x(I, p, c) = 0$ for every color $c \in [\chi] \setminus \{\delta(e_1)\}$,
		\item $x(I, q, c) = 0$ for every page $q \in [\ell] \setminus \{p\}$ and color $c \in [\chi]$,
		\item $v(I, p) = v$, and
		\item $v(I, q) = \bot$ for every page $q \in [\ell] \setminus \{p\}$.\label{dp:queue-layouts-pages-colors-xp-base-end}
	\end{enumerate}
	\smallskip
	Observe that the above conditions encode the assignment of the edge $e_1$ to the page $p$.
	
	Now consider the general case with $e(I) = e_i = uv$, $u \prec v$, for some $i > 1$. 
	We set $D[I] = 1$ if and only if there exists a state $I'$ with $e(I') = e_{i - 1}$, $D[I'] = 1$, and a page $p \in [\ell]$ such that all of the following conditions are satisfied.
	\begin{enumerate}[G1.]
		\item $x(I, p, \delta(e_i)) = x(I', p, \delta(e_i)) + 1$\label{dp:queue-layouts-pages-colors-xp-general-start},
		\item $x(I, p, c) = x(I', p, c)$ for every color $c \in [\chi] \setminus \{\delta(e_i)\}$,
		\item $x(I, q, c) = x(I', q, c)$ for every page $q \in [\ell] \setminus \{p\}$ and color $c \in [\chi]$,
		\item $v(I', p) = \bot$ or $v(I', p) = w$ for some $w \in V$ with $w \preceq v$\label{dp:queue-layouts-pages-colors-xp-general-no-nesting}, 
		\item $v(I, p) = v$, and
		\item $v(I, q) = v(I', q)$ for every page $q \in [\ell] \setminus \{p\}$.\label{dp:queue-layouts-pages-colors-xp-general-end}
	\end{enumerate}
	\smallskip
	As before, these conditions encode the assignment of the edge $e_i$ to the page $p$.
	In particular, Condition~G\ref{dp:queue-layouts-pages-colors-xp-general-no-nesting} ensures that the edge $e_i$ does not nest another edge on page $p$; recall \Cref{lem:queue-layouts-pages-colors-nesting}. 
	
	We now show that the above conditions preserve feasibility of the states.
	\begin{claim}
		\label{claim:queue-layouts-pages-colors-xp-correctness}
		Let $I$ be a state.
		We have $D[I] = 1$ if and only if $I$ is feasible.
	\end{claim}
	\begin{claimproof}
		We show the claim using induction on (the order of) $e(I)$.
		
		As our base case, we consider states $I$ with $e(I) = e_1$.
		Recall for the following arguments Conditions~B\ref{dp:queue-layouts-pages-colors-xp-base-start}--B\ref{dp:queue-layouts-pages-colors-xp-base-end}.
		Assume that $I$ is feasible.
		Thus, by definition, there exists an $\ell$-page queue layout $\langle \prec, \sigma_1\rangle$ of $G_1$ such that there are precisely $x(I, p, c)$ edges of color $c$ assigned to page $p$, for every $c \in [\chi]$ and $p \in [\ell]$, and $v(I, p)$ is the rightmost vertex which is incident to an edge on page $p$, for every $p \in [\ell]$.
		Since $G_1$ contains precisely the edge $e_1$, the page $p = \sigma_1(e_1)$ witnesses that all conditions are satisfied.
		Thus, we have $D[I] = 1$.
		Conversely, assume $D[I] = 1$.
		Let $p \in [\ell]$ be the page that satisfies Conditions~B\ref{dp:queue-layouts-pages-colors-xp-base-start}--B\ref{dp:queue-layouts-pages-colors-xp-base-end} and consider the linear layout $\langle \prec, \sigma_1 \rangle$ with $\sigma_1(e_1) = p$.
		Since $G_1$ contains precisely the edge $e_1$, this layout witnesses the feasibility of $I$.
		
		For the inductive step, we assume that the statement holds for all states $I$ with $e(I) = e_{j}$ for $j < i$.
		Recall for the following arguments Conditions~G\ref{dp:queue-layouts-pages-colors-xp-general-start}--G\ref{dp:queue-layouts-pages-colors-xp-general-end}.
		Assume that $I$ is feasible and consider the $\ell$-page queue layout $\langle \prec, \sigma_i\rangle$ of $G_i$ witnessing this.
		Observe that the layout $\langle \prec, \sigma_{i - 1}\rangle$ with $\sigma_{i - 1}(e_j) = \sigma_i(e_j)$ for all $e_j \in E$ with $j < i$ is an $\ell$-page queue layout of $G_i - \{e_i\} = G_{i - 1}$.
		In particular, the layout witnesses the feasibility of the corresponding state $I'$.
		Observe that we consider $I'$ when evaluating $D[I]$ as $e(I') = e_{i - 1}$.
		Moreover, in terms of the values $x(\cdot, \cdot, \cdot)$, the two states differ precisely in $x(\cdot, \sigma_i(e_i), \delta(e_i))$, and we have $x(I, \sigma_i(e_i), \delta(e_i)) = x(I', \sigma_i(e_i), \delta(e_i)) + 1$.
		Regarding the $v(\cdot, \cdot)$-values, the two states differ only on the page $\sigma_i(e_i)$.
		Since $e_j < e_i$ for every edge $e_j$ of $G_{i - 1}$, we have that $v(I', \sigma_i(e_i)) \preceq v$ for $e_i = uv$; recall \Cref{lem:queue-layouts-pages-colors-nesting}.
		Consequently, we satisfy Conditions~G\ref{dp:queue-layouts-pages-colors-xp-general-start}--G\ref{dp:queue-layouts-pages-colors-xp-general-end} for state $I'$ by choosing $p = \sigma_i(e_i)$ and thus have $D[I] = 1$.
		Finally, assume $D[I] = 1$ and consider the state $I'$ and the page $p \in [\ell]$ that we used to satisfy Conditions~G\ref{dp:queue-layouts-pages-colors-xp-general-start}--G\ref{dp:queue-layouts-pages-colors-xp-general-end}.
		Since $e(I') = e_{i - 1}$ and $D[I'] = 1$, $I'$ is feasible by induction hypothesis.
		Consider the $\ell$-page queue layout $\langle \prec, \sigma_{i - 1}\rangle$ of $G_{i - 1}$ witnessing feasibility of $I'$ and extend it by assigning the edge $e_i$ to page $p$.
		We observe that the resulting layout $\langle \prec, \sigma_{i}\rangle$ is an $\ell$-page queue layout of $G_{i - 1} + \{e_i\} = G_i$.
		To see why this is the case, assume that two edges nest on some page and observe that $e_i$ must be one of them.
		Let the other edge be $e_j = xy$ with $x \prec y$.
		Since we consider the edges in lexicographic order, we have $e_j < e_i$.
		Thus, we have $x \prec u \prec v \prec y$ ($u \prec x \prec y \prec v$ would imply $e_i < e_j$).
		Since $D[I] = 1$, $I$ satisfies Condition~G\ref{dp:queue-layouts-pages-colors-xp-general-no-nesting}.
		This implies $v(I', \sigma_i(e_i)) = \bot$ or $v(I', \sigma_i(e_i)) = w$ for some $w \in V$ with $w \preceq v$.
		However, this contradicts the above-established order $x \prec u \prec v \prec y$ as the endpoint of the edge $e_j$ on page $\sigma_i(e_i)$ is to the right of $v$.
		Thus, $e_j$ and $e_i$ cannot nest and $\langle \prec, \sigma_{i}\rangle$ is an $\ell$-page queue layout of $G_i$.
		The remaining conditions ensure consistency between the two states.
		Hence, $\langle \prec, \sigma_{i}\rangle$ witnesses feasibility of $I$.
	\end{claimproof}
	
	Once the DP-table $D$ has been computed, we evaluate every state $I$ with $e(I) = e_m$ and $D[I] = 1$ and check whether it adheres to our budget constraints, i.e., whether we have $x(I, p, c) \leq b_p(c)$ for every page $p \in [\ell]$ and color $c\in [\chi]$.
	If we can find such a state, then we can conclude, by \Cref{claim:queue-layouts-pages-colors-xp-correctness}, that \instance is a positive instance.
	On the other hand, if \instance is a positive instance, the corresponding solution witnesses feasibility for some state $I$ with $e(I) = e_m$, which implies $D[I] = 1$ by \Cref{claim:queue-layouts-pages-colors-xp-correctness}.
	Call this last step the \emph{postprocessing} step.
	Combining all, we conclude that our DP-algorithm is correct and we now argue its running time.
	
	\proofsubparagraph*{Running Time.}
	We first bound the size of the DP-table $D$.
	The graph $G$ contains $m$ edges and the value $x(I, p, c)$ ranges between $0$ and $m$ for every page $p \in [\ell]$ and color $c \in [\chi]$.
	Moreover, there are $n + 1$ possible values for $v(I, p)$ for every page $p \in [\ell]$; recall the dummy value $\bot$.
	Hence, there are $\BigO{m \cdot {(m + 1)}^{\ell \cdot \chi} \cdot (n + 1)^{\ell}}$ different states $I$, which also bounds the size of $D$.
	Evaluating a single state $I$ takes $\BigO{\ell\cdot n}$ time, since we just need to examine all potential page assignments for the edge $e(I)$ and evaluate the resulting state $I'$, for which only the rightmost vertex is still unspecified.
	Thus, filling the entire table $D$ takes $\BigO{\ell\cdot n \cdot m \cdot {(m + 1)}^{\ell \cdot \chi} \cdot (n + 1)^{\ell}} \in \BigO{\ell\cdot {(m + 1)}^{(\ell \cdot \chi) + 1} \cdot (n + 1)^{\ell + 1}}$ time.
	Since this dominates the time for the pre- and postprocessing steps, this also bounds the overall running time of our algorithm.
	Finally, we can obtain an admissible $\ell$-page queue layout by standard backtracking techniques (if there exists one).
\end{prooflater}

\subsection{Towards Completing the Complexity Landscape}
\label{sec:queue-layouts-small-conflict-graph}
In this section, we consider the remaining parameter combinations depicted in \Cref{fig:overview}.
First, we parameterize by the number of pages $\ell$, the number of colors $\chi$, and the maximum size \conCompSize of a connected component in the conflict graph \queueConflict.
There, we can apply the same proof strategy as for \Cref{thm:stack-layouts-pages-colors-components-size} (when considering nestings instead of crossings).
More concretely, we partition the connected components of $\queueConflict$ into $g(\chi, \conCompSize)$ different types such that components of the same type are equal up to a (color-preserving) isomorphism.
Then, we formulate an ILP with $h(\ell, \chi, \conCompSize)$ variables that constrains the number of components of each type that can use the same page assignment.
Together with \Cref{prop:ilp-fpt}, this implies:
\begin{theorem}
    \label{thm:queue-layouts-pages-colors-components-size}
    \capacitatedQueueLong is \FPT\ in the number of pages $\ell$, number of colors $\chi$, and the maximum size $\conCompSize$ of a connected component of $\queueConflict[G, \prec]$.
\end{theorem}

Finally, we parameterize by the number \conComp of connected components in the conflict graph~\queueConflict and their maximum size $\conCompSize$.
Consider an instance \instance of \capacitatedQueue.
Since $\queueConflict$ contains one vertex for each edge of $G$ and we can assume without loss of generality that there are no isolated vertices in $G$, we have $n \leq 2m \leq 2\cdot \conComp \cdot \conCompSize$.
With this we can show that we can restrict ourselves to a bounded number of pages and colors, i.e., obtain a kernel, which allows us to brute-force $\instance$:
\begin{restatable}\restateref{thm:queue-small-conflict-graph}{theorem}{theoremQueueSmallConflict}
    \label{thm:queue-small-conflict-graph}
    \capacitatedQueueLong can be solved constructively in $2^{\BigO{(\conComp \cdot \conCompSize)^2\cdot\log(\conComp \cdot \conCompSize)}} + \BigO{\ell \cdot \chi}$ time (which is \FPT\ in the number \conComp of connected components of $\queueConflict[G, \prec]$ and their maximum size \conCompSize).
\end{restatable}
\begin{prooflater}{ptheoremQueueSmallConflict}
    Let $\instance = \instanceLongQueue$ be an instance of \capacitatedQueue.
    Consider the conflict graph $\queueConflict$.
    It has at most $\conComp \cdot \conCompSize$ vertices and since it contains precisely one vertex for every edge of $G$, we get $m \leq \conComp \cdot \conCompSize$.
    Moreover, we can assume without loss of generality that there are no isolated vertices in $G$, as they can be removed without affecting the existence of a solution.
    Thus, we have $n \leq 2m \leq 2 \cdot \conComp \cdot \conCompSize$.

    Since $m \leq \conComp \cdot \conCompSize$, we can also assume that $\chi \leq  \conComp \cdot \conCompSize$ holds.
    Otherwise, we can remove unused colors and renumber them from $1$ to $\conComp \cdot \conCompSize$ to obtain an equivalent instance where the above inequality holds.
    By similar arguments, we can also assume that $b_p(c) \leq m \leq \conComp \cdot \conCompSize$ holds for every page $p \in [\ell]$ and color $c \in [\chi]$.
    Consequently, there are at most ${(m + 1)}^{\chi}$ different budget functions $b_p$.
    If \instance has more than $m$ pages with the same budget vector $b_p$, then we can remove one page from the instance without affecting the existence of a solution.
    This is because, by the pigeonhole principle, in every solution one of these pages will remain empty so we can convert any solution into one where the ``removed pages'' are empty.
    Thus, it suffices to consider at most $m\cdot {(m + 1)}^{\chi}$ different pages $p$ and we can find them after a $\BigO{\ell\cdot \chi}$-time preprocessing step.

    This allows us to brute-force all possible page assignments, which are $2^{\BigO{(\conComp \cdot \conCompSize)^2\cdot\log(\conComp \cdot \conCompSize)}}$-many.
    For each page assignment $\sigma$, we can check in polynomial time (in $m \leq \conComp \cdot \conCompSize$) whether they correspond to a solution \linL.
    The statement now follows.
\end{prooflater}

\section{Generalizations and Concluding Remarks}
In this paper, we have introduced the concept of a capacitated linear layout and studied the complexity of computing an $\ell$-page capacitated stack/queue layout with given spine order~$\prec$.
Although we only considered unweighted graphs in this paper, our algorithmic results can be extended to edge-weighted graphs (when the largest edge-weight is a parameter).

One prominent open question is settling the complexity of \capacitatedQueue when parameterized by the number of colors and the size of the largest connected component in the conflict graph \queueConflict; recall \Cref{fig:overview}.
Furthermore, recent work studied the parameterized complexity of computing an $\ell$-page stack and queue layout with respect to several graph parameters (even when $\prec$ is not part of the input); see for example~\cite{DFGS.LLR.2025}.
Extending these results to the capacitated setting is an interesting direction for future work.
A further natural direction is other variants of capacitated linear layouts, such as capacitated \emph{mixed} linear layouts.

\bibliography{references}

\onlyShort{
\newpage
\appendix

\crefalias{section}{appendix}
\crefalias{subsection}{appendix}

\section{Omitted Details from \Cref{sec:stack-layouts}}
\label{app:stack-layouts}

\subsection{Omitted Proofs from \Cref{sec:stack-layouts-components}}
\label{app:stack-layouts-components}

\theoremStackComponentsFPT*
\label{thm:stack-layouts-components-fpt*}
\ptheoremStackComponentsFPT

\subsection{Omitted Details from \Cref{sec:stack-layouts-components-size}}
\label{app:stack-layouts-components-size}

\stackComponentSizeGadgets

\lemmaStackComponentsSizeClause*
\label{lem:stack-layouts-components-size-clause*}
\plemmaStackComponentsSizeClause

\theoremStackComponentsSize*
\label{thm:stack-layouts-components-size*}
\ptheoremStackComponentsSize

\subsection{Omitted Details from \Cref{sec:stack-layouts-colors}}
\label{app:stack-layouts-colors}

\lemmaStackColorsXP*
\label{lem:stack-layouts-colors-xp*}
\plemmaStackColorsXP

\stackColorsNumberGadget

\lemmaStackColorsWOne*
\label{lem:stack-layouts-colors-w-one*}
\plemmaStackColorsWOne

\subsection{Omitted Proofs from \Cref{sec:stack-layouts-colors-components-size}}
\label{app:stack-layouts-colors-components-size}

\lemmaStackColorsComponentsTypes*
\label{lem:stack-layouts-colors-components-size-types*}
\plemmaStackColorsComponentsTypes

\theoremStackColorsComponents*
\label{thm:stack-layouts-pages-colors-components-size*}
\ptheoremStackColorsComponents

\section{Omitted Details from \Cref{sec:queue-layouts}}
\label{app:queue-layouts}

\subsection{Omitted Details from \Cref{sec:queue-layouts-two-pages}}
\label{app:queue-layouts-two-pages}

\queueLayoutsComponentsSizeHardnessDetails
\theoremQueueComponentsSize*
\label{thm:queue-layouts-components-size*}
\ptheoremQueueComponentsSize

\theoremQueueThreePagesConnected*
\label{thm:queue-layouts-three-pages-connected*}
\ptheoremQueueThreePagesConnected

\subsection{Omitted Proofs from \Cref{sec:queue-layouts-monochromatic}}
\label{app:queue-layouts-monochromatic}

\theoremQueueMonochromatic*
\label{thm:queue-layouts-monochromatic*}
\ptheoremQueueMonochromatic

\theoremQueueColorsConnected*
\label{thm:queue-layouts-two-colors-connected-conflict*}
\ptheoremQueueColorsConnected

\subsection{Omitted Proofs from \Cref{sec:queue-layouts-pages-colors}}
\label{app:queue-layouts-pages-colors}

\lemmaQueuePagesColorsNesting*
\label{lem:queue-layouts-pages-colors-nesting*}
\plemmaQueuePagesColorsNesting

\theoremQueuePagesColors*
\label{thm:queue-layouts-pages-colors-xp*}
\ptheoremQueuePagesColors

\subsection{Omitted Proofs from \Cref{sec:queue-layouts-small-conflict-graph}}
\label{app:queue-layouts-small-conflict-graph}

\theoremQueueSmallConflict*
\label{thm:queue-small-conflict-graph*}
\ptheoremQueueSmallConflict

}

\end{document}